\documentclass[acmsmall,nonacm,screen,natbib=false]{acmart}

\usepackage[utf8]{inputenc}
\usepackage[T1]{fontenc}

\usepackage{comment}

\usepackage{url}
\usepackage{hyperref}

\RequirePackage[
  datamodel=acmdatamodel,
  url=false,
  maxcitenames=2,
  maxbibnames=9,
  giveninits,
  ]{biblatex}
\AtBeginBibliography{\footnotesize}

\usepackage[raster]{tcolorbox}
\usepackage[noline,noend]{algorithm2e}
\usepackage{graphicx}
\usepackage{tikz}
\usetikzlibrary{graphs, quotes}

\usepackage{multirow}
\RequirePackage{multicol}
\usepackage{booktabs}
\usepackage{subcaption}

\usepackage{environ}

\RequirePackage{amsmath, amsthm, mathtools}
\input{mathdefs}
\mathcode`:="603A          

\newenvironment{boxes}[1]
{\begin{tcbraster}[raster columns=#1, raster equal height, size=small,
  space to upper,
  colframe=white,
  colback=white,
  coltitle=black,
  valign=center,
  halign=center,
  ]}
{\end{tcbraster}}

\newcommand\axiom[2]{\ensuremath{\textsc{#1}_\textsc{#2}}}
\newcommand\cons[1]{\textsc{#1}}

\auxfun{inc}
\auxfun{dec}
\auxfun{add}
\auxfun{rmv}
\auxfun{enq}
\auxfun{deq}
\auxfun{push}
\auxfun{pop}
\auxfun{val}
\auxfun{cbcast}
\auxfun{deliver}

\auxfun{proc}
\auxfun{pr}
\auxfun{op}
\auxfun{res}
\auxfun{re}
\auxfun{obj}
\auxfun{tentative}
\auxfun{commited}

\mkrel{vis}
\mkrel{po}
\mkrel{ser}
\mkrel{ar}
\mkrel{hb}
\mkrel{mhb}
\mkrel{phb}
\mkrel{tran}
\semfun{R}
\semfun{S}
\semfun{C}

\newtheorem{definition}{Definition}
\newtheorem{proposition}{Proposition}
\newtheorem{theorem}{Theorem}
\newtheorem{remark}{Remark}

\begin{document}

\title{A Framework for Consistency Models in Distributed Systems}
\subtitle{Time, Visibility, Order, Convergence, Basic Axioms, and a New Trilemma}
\author{Paulo Sérgio Almeida}
\email{psa@di.uminho.pt}
\orcid{0000-0001-7000-0485}
\affiliation{%
  \institution{INESC TEC \& University of Minho}
  \country{Portugal}
}

\begin{abstract}
Consistency models define possible outcomes when concurrent processes use a
shared abstraction. Classic models explain the local history at each process
as a serial execution where effects of remote operations are interleaved in
some order.  Some frameworks allow types with sequential specifications.
But many useful abstractions like some CRDTs do not have sequential
specifications; other abstractions have sequential specifications but cannot
be explained by serial executions.  Some frameworks allow concurrent
specifications but focus on eventual consistency and cannot describe
non-converging models.  Most models are timeless, so they can be used for
distributed histories with no order information about operations from
different processes.  As we observe in this article, many classic models,
such as PRAM and causal memory, allow histories that can only be explained
by physically impossible causality cycles. Models that use a happens-before
relation assume it to be a strict partial order, preventing visibility
cycles. But forbidding all cycles is overly restrictive: it does not allow
histories from synchronization-oriented abstractions with partial concurrent
specifications, such as a barrier or a consensus object.

We define a general axiomatic timeless framework for asynchronous distributed
systems, together with well-formedness and consistency axioms. It unifies and
generalizes the expressive power of current approaches, simultaneously in
several dimensions.
1) It combines classic serialization per-process with a global visibility,
in abstract executions and semantic functions.
2) It defines a physical realizability well-formedness axiom to prevent
physically impossible causality cycles, while allowing possible and useful
visibility cycles, namely to allow synchronization-oriented abstractions,
allowing happens-before to be a preorder.
3) Allows adding time-based constraints, from a logical or physical clock,
either partially or totally ordered, in an optional and orthogonal way, while
keeping models themselves timeless.
4) It simultaneously generalizes from memory to general abstractions, from
sequential to both sequential and concurrent specifications, either total or
partial, and beyond serial executions.
5) Defines basic consistency axioms: monotonic visibility, local visibility,
and closed past. These are satisfied by what we call serial consistency, the
model explained by serial execution, but can be used as building blocks for
novel consistency models with histories not explainable by any serial
execution.
6) Revisits classic pipelined and causal consistency, revealing weaknesses
in previous axiomatic models for PRAM and causal memory.
7) Introduces convergence and arbitration as safety properties for consistency
models, departing from the use of eventual consistency, which conflates
safety and liveness.
8) Formulates and proves the CLAM theorem for asynchronous distributed systems,
which essentially says that any wait-free implementation of practically all
data abstractions cannot simultaneously satisfy Closed past, Local visibility,
Arbitration, and Monotonic visibility.  From it results the CAL trilemma:
which of Closed past, Arbitration, or Local visibility to forgo. While
technically incomparable, the CLAM theorem is practically stronger than the
CAP theorem, as it allows reasoning about the design space and possible
tradeoffs in highly available partition tolerant systems.
9) Presents a new taxonomy of consistency models that shows how the introduced
axioms can be combined, clearly showing the tradeoffs resulting from the CAL trilemma.
\end{abstract}



\maketitle

\section{Introduction}

A consistency model describes what are the possible outcomes when several
concurrent processes use a shared abstraction. For a long time, the focus was
on memory, a set of registers with read and write operations, for which most
consistency models where developed.

When multiprogramming started in uniprocessors, by operating systems
time-slicing processor time by preemptive multitasking, there was even no
discussion of consistency models. Memory acted as it ``obviously'' should:
if a process wrote to some memory location, the value would be immediately
available to be read by any other process. Each write and read would be
atomic, and whatever happened would be equivalent to the actions (read or
write) by all processes being interleaved in some global sequence, respecting
the program order at each process. This guarantee is
what~\textcite{DBLP:journals/tc/Lamport79} called \emph{Sequential
Consistency}. It became the goal to be achieved in multiprocessors or
distributed systems, to make programmer life not too hard;
concurrent programming is hard enough even under this intuitive model.

For some time this guarantee would hold (e.g., a simple bus-based
multiprocessor with no per-processor cache), but as hardware became more
sophisticated (e.g., with cache hierarchies and write buffers)
guaranteeing sequential consistency would hinder performance. Memory started
behaving in possibly non-intuitive ways. This lead to the need to describe
what could possible happen, and the so called \emph{weak consistency models} arose.

Traditional models and frameworks for consistency
models~\cite{DBLP:journals/computer/AdveG96}, focused on memory, tended to
use an operational approach, making use of implementation related terms such
as cache or write buffer.
They describe what is forbidden or allowed to occur, e.g., stating what
relaxations of different orders between reads and writes are possible, or if a
write must appear atomic, with descriptions such as ``whether they allow a
read to return the value of another processor’s write before all cached copies
of the accessed location receive the invalidation or update messages generated
by the write''~\cite[p. 11]{DBLP:journals/computer/AdveG96}.
Operational approaches make it difficult to compare different models, or
to reason about the design space for new models. As an example,
\textcite{DBLP:conf/isca/Gharachorloo98} provides an operational definition
for \emph{Processor Consistency} which allows outcomes slightly different 
than the original by~\textcite{goodman1991cache}, as discussed
by~\textcite{DBLP:conf/spaa/AhamadBJKN93}.

The alternative axiomatic approach~\cite{DBLP:journals/toplas/Misra86}, more
formal and declarative, does not resort to implementation concepts, but
defines axioms regarding what possible histories are allowed for a given
model. A history is explained by some existentially assumed abstract execution.
Some frameworks~\cite{DBLP:conf/spaa/AhamadBJKN93,DBLP:conf/europar/BatallerB97}
use the concept of \emph{serial view}: a serialization of a subset of
operations satisfying some ordering constraints, representing a serial
execution of those operations. Orderings are defined based on program order
and a writes-to order.

Based on serial views, \textcite{DBLP:journals/jacm/SteinkeN04} defined a
framework which allows different consistency models to be obtained as a
combination of several properties, and be ordered by strength in a lattice of
models. The framework allows describing all classic models:
Slow~\cite{DBLP:conf/icdcs/HuttoA90}, Cache~\cite{goodman1991cache},
PRAM~\cite{lipton1988pram}, Processor~\cite{goodman1991cache},
Causal~\cite{DBLP:journals/dc/AhamadNBKH95}, and
Sequential~\cite{DBLP:journals/tc/Lamport79} consistency. In all these models,
each process explains all local reads from some serial execution respecting
program order for local operations, and some order of writes from other
processes. The weakest such model is
what~\textcite{DBLP:conf/europar/BatallerB97} call \emph{Local Consistency}.

The first dimension of generalization was defining consistency models for
higher level abstractions, concurrent shared objects, where memory is just a
special case. The most well-known specific model is
Linearizability~\cite{DBLP:journals/toplas/HerlihyW90}. Defining a framework
to allow different consistency models while going beyond memory brings some
differences. The main one is that a query is in general explained by several
updates, while in a serial view for memory a read is explained by a single write
(the last one before the read). Interestingly, going beyond memory, as
\textcite{DBLP:conf/ppopp/PerrinMJ16} did for causal consistency, revealed an
existing problem in memory-specific models and the above framework using
serial views: they use an existentially assumed writes-to relation, where each
read is related to a single write, but the last write in a serial view may be
another one. More general frameworks beyond memory (such as ours) solve this problem.

Local consistency as the weakest model in the lattice of models is motivated
by the intuition that any process must be able to explain what it reads by
some serial execution that contains all operations from the process itself in
program order. Unfortunately, while this assumption is reasonable for
shared-memory multiprocessors, the same cannot be said for distributed
systems.  Programming a distributed system, through explicit message passing,
and explicit decisions of what to store at each node, allows many
possibilities in terms of defining what should be the set of possible outcomes
when processes use an abstraction.

Aiming for scalability of queries and tolerance to a server being down led to
replicating data over several servers. But strong models such as
linearizability cannot be ensured while guaranteeing availability in the
presence of a network partition. This became known as the CAP
conjecture~\cite{DBLP:conf/hotos/FoxB99} and
theorem~\cite{DBLP:journals/sigact/GilbertL02}).
Achieving availability is a motivation for optimistic
replication~\cite{DBLP:journals/csur/SaitoS05}, where weaker
guarantees are provided, to ensure availability even under network partitions.
Bayou~\cite{DBLP:conf/sosp/TerryTPDSH95} is an example of this approach.
However, it allows histories that cannot be explained by any serial execution
with the local operations in program order, even if the remote effects are
interleaved in any order. I.e., it does not respect the generalization of
local consistency to arbitrary data types -- what we call \emph{Serial Consistency}.
Other more modern examples are OpSets~\cite{DBLP:journals/corr/abs-1805-04263}
or ECROs~\cite{DBLP:journals/pacmpl/PorreFPB21}.
No classic consistency model accepts their behaviors, and almost no framework
for consistency models beyond memory is able to express such outcomes.

In the above examples, as is usually the case, the shared data abstraction is
defined through a sequential specification. This is appealing but mat be
limiting. An important principled approach to obtain highly-available shared
abstractions are the so called CRDTs -- \emph{Conflict-free Replicated Data
Types}~\cite{DBLP:conf/sss/ShapiroPBZ11}. They allow defining shared objects,
such as counters, sets, or lists, where queries can be answered immediately, from
the local state, and updates propagated asynchronously, opportunistically, as
desired or network conditions allow.
While network partitions may have a transient impact on state growth in the
operation-based approach~\cite{10.1145/3695249}, state-based CRDTs are
completely immune to network partitions. CRDTs are important in practice and
have been used successfully in the industry. However, several CRDTs, of which
the Observed-Remove Set (ORSet) and the Multi-Value Register (MVRegister) are
well known examples, do not have sequential specifications. They rely on
concurrent specifications that rely on the causal past. An MVRegister
corresponds to the behavior of the seminal
Dynamo~\cite{DBLP:conf/sosp/DeCandiaHJKLPSVV07} from Amazon, where a read returns
the set of the most recent concurrent writes.

\subsection{Scope: distributed systems and distributed histories}

Consistency models can be roughly divided into those for shared
memory multiprocessors and those for distributed systems (even if for
distributed shared memory).

Memory consistency models for shared-memory multiprocessors aim to address
architecture specific features (e.g., x86-TSO~\cite{DBLP:journals/cacm/SewellSONM10})
or what are the possible behaviors of concurrent programs in some language (e.g.,
C++~\cite{DBLP:conf/pldi/BoehmA08,DBLP:conf/popl/BattyOSSW11}).
Either way, they face a combination of aspects which together bring enormous
complexity.

Models for high level programming languages face the problem of how to define
compiler optimizations that, not only preserve the semantics of a sequential
program, but do not cause havoc when several processes are involved. But even
without compiler optimizations, and considering assembly language programs to
be directly executed, there is still the issue of data dependencies and
control dependencies, in addition to processors features motivated by
performance: instruction reordering, speculative execution, local cache
hierarchies, and write buffers.

One of the most difficult and undesirable phenomenon is that of
out-of-thin-air (OOTA) results: values that are only explained by a circular
reasoning, which would also explain any other value.
Consider the three multi-threaded programs in Figure~\ref{fig:programs-oota}, 
in which two threads read and write shared variables x and y, and also use local
variables l1 and l2.
Under a weak model, the program in Figure~\ref{fig:oota-possible} allows both
reads to return 42; this is explained by reordering the write to x before the
read of y. But a similar result in Figure~\ref{fig:oota-invalid} can only be
explained by a circular reasoning: the value each thread reads, and uses for
the write, depends on what the other thread wrote, but no thread writes any
particular value, making 42 (or any other value) come from nowhere.

\newcommand\threads[2]
{
  \begin{tabular}{lcl}
  \begin{tabular}{l}
    #1
  \end{tabular}
    & $\parallel$ &
  \begin{tabular}{l}
    #2
  \end{tabular}
  \end{tabular}
}

\begin{figure*}
  \small
  \begin{boxes}{3}
    \begin{tcolorbox}
      \threads{l1 = x; \\ y = l1;}{l2 = y; \\ x = 42;}
      \subcaption{l1 = l2 = 42 is possible}
      \label{fig:oota-possible}
    \end{tcolorbox}
    \begin{tcolorbox}
      \threads{l1 = x; \\ y = l1;}{l2 = y; \\ x = l2;}
      \subcaption{l1 = l2 = 42 is OOTA}
      \label{fig:oota-invalid}
    \end{tcolorbox}
    \begin{tcolorbox}
      \threads{if (x) \\ \ \ \ \ y = 1;}{if (y) \\ \ \ \ \ x = 1;}
      \subcaption{x = y = 1 is OOTA}
      \label{fig:oota-sc-drf}
    \end{tcolorbox}
  \end{boxes}
  \caption{Three multi-threaded programs; x and y are shared variables, l1 and
  l2 local variables. (a) reads can return 42,
  explained by the reordering of x = 42 and l2 = y. (b) Any value, like 42,
  comes out-of-thin-air, only explained by a circular reasoning. (c) Assuming
  sequential consistency the program has no races; it should not accept the
  outcome where branches run.}
  \label{fig:programs-oota}
\end{figure*}

While these examples are easy to understand, it is extremely
difficult~\cite{DBLP:conf/pldi/BoehmD14} to define exactly what an OOTA result
is. OOTA results can determine which code branches execute, and which writes are
performed. In Figure~\ref{fig:oota-sc-drf}, any condition being true (1)
could only happen if any write occurs; but these can only occur if a condition
is true.
A goal of consistency models for programming languages has been to ensure
sequential consistency (SC) for data race free (DRF) programs
(DRF-SC)~\cite{DBLP:conf/isca/AdveH90}, i.e., if a program is DRF assuming SC,
then it will satisfy SC.
The program in Figure~\ref{fig:oota-sc-drf} is DRF assuming SC (no write
occurs) and so it should not produce the OOTA result x = y = 1, invalid under SC.

This article does not address shared-memory multiprocessors or concurrent
languages, but focuses on consistency models for distributed systems. In
these, the ``input'' to be checked for outcomes is not a program but a
distributed history. A distributed history can be thought of as the recording
of what operations, with respective results, each process has performed over
time. For traditional models (PRAM, causal memory, sequential consistency) a
distributed history has been defined as a collection of sequences of
operations, with the operations of each process totally ordered, but no
information about the relative order of operations from different processes.
Occasionally, time information is also included, as for linearizability.

Focusing on distributed histories sidesteps most difficult problems discussed
above. There is no program, no data dependencies, no control dependencies.
For example, the situations in Figure~\ref{fig:oota-possible} and
Figure~\ref{fig:oota-invalid} are indistinguishable: the history would only have
the values (e.g., 42) themselves, not where they came from.

Moreover, the history represents, for each process, the actual order in which
operations occurred in time. If we think of a program being optimized and
instructions being reordered, the history is the recording post any
transformation, even inside an hypothetical processor. The history represents
the order in which each process interacted with other components of the
distributed system. Consistency models should aim to identify 1) if the history
can be explained by some pattern of information propagation, discarding
``impossible histories'' (something which is often forgotten); 2) if the
history can be explained while respecting specific constraints from the
consistency model.

For accessing whether a history is physically possible, the central aspect is
the assumption of precedence in physical time between operations from the same
process, that the actual execution respects the program order: the next
operation by a process only starts, and can produce effects on the rest of the
system, after the previous operation has returned a result, which cannot
depend on those effects. This will reject the history when x = y = 42,
corresponding to both Figure~\ref{fig:oota-possible} and Figure~\ref{fig:oota-invalid}.
Programs frequently use query results to compute arguments to subsequent updates.
If histories as those were accepted, they would be allowing OOTA for common programs.
The assumption ensures that, whatever the (unknown) data dependencies in the
(unknown) program that produced the history, it can be explained by physical
propagation of information, without circular reasoning, without the danger
that it can only be explained by OOTA results.
Consistency models themselves also tend to be conservative. An example is causal
consistency, which deals with potential causality, ensuring that any actual
dependency is met by an accepted history.

While distributed histories are simpler to deal with than concurrent programs,
models for distributed systems focus on something that is not an issue for
shared-memory models: convergence. These typically assume that there is ``a''
memory and that eventually each location in ``the'' memory ends up with some
value. This can be seen in the abstract machine for
x86-TS~\cite{DBLP:journals/cacm/SewellSONM10}: values to be stored go through
write buffers but end up in ``the'' shared memory.

Consistency models for distributed systems are suitable to systems in which
each location is replicated in more than one place (from one replica
per process, in peer-to-peer architectures, to one replica per server, with
processes being clients to those servers). In this setting it is possible that
replicas diverge forever, if nothing is done to prevent it. Traditional models
for distributed systems weaker than sequential consistency (such as PRAM and
causal memory) do not ensure convergence.

\subsection{Contributions}

This article presents a unifying framework and axioms for consistency
models for asynchronous distributed systems.
It generalizes current frameworks to address simultaneously the aspects
described above (going beyond memory, going beyond sequential specifications,
going beyond histories that can be explained by serial executions), as well as
other relevant issues, such as how to deal with time and how convergence can
be achieved.
It makes several contributions spanning many different but related aspects which,
together, advance significantly the state of the art in consistency models.

\paragraph{Novel framework} Introduces a new framework for
consistency models, timeless (no notion of ``real time'') and abstract
(independent from any implementation details), based on the combination of a
global visibility relation and per-process serializations. While each has been
individually used in previous frameworks, their combination in our
framework is a novel and powerful feature.
Visibility is defined in a way to avoid notational clutter and allow elegant
formulations of constraints, unlike in some previous approaches (e.g.,
\cite{DBLP:conf/popl/BurckhardtGYZ14}), which end up notationally heavy.
Here, the essential concepts and differences become clearly visible.

\paragraph{Time and consistency models} Clarifies the role of time and
happens-before, distinguishing between physical, messaging and semantic
happens-before (the last being the relevant one for consistency models). Shows
how knowledge about time, through logical or physical clocks, either totally or
partially ordered, can be used to
prune the space of candidate visibility relations, to define \emph{time valid
executions}.  This allows an orthogonal combination of time constraints with
timeless consistency models.

\paragraph{Physical realizability axiom}
While the framework itself is timeless, it includes a physical realizability
axiom, derived from the precedence in physical time between operations from
the same process. This axiom is missing from classic axiomatic specifications
such as PRAM and causal memory, which allow histories that can only be explained
by circular reasoning and out-of-thin-air results. Our axiom overcomes the problem
in a unified way, freeing specific models from having to introduce specific ad
hoc axioms. Importantly, generalizing frameworks where happens-before is a
strict partial order, the axiom does not overconstrain possible visibility, and
allows describing synchronization-oriented abstractions such as barriers,
through partial concurrent specifications, where visibility contains
cycles and happens-before is a preorder.

\paragraph{Dimensions of expressivity} The framework generalizes the
representable consistency models simultaneously in three dimensions: 1) from
read/write registers to general abstractions; 2) allows both sequential and
concurrent specifications; 3) allows describing models with either final
executions, or that involve (conceptual) re-execution, where the past may be
``rewritten''.  Most current approaches either: consider concurrent specifications
in the context of eventual
consistency~\cite{DBLP:journals/ftpl/Burckhardt14,DBLP:conf/popl/BurckhardtGYZ14},
resorting to visibility and arbitration (a global total order), without
per-process serializations; or resort only to per-process serializations,
without visibility, being limited to sequential
specifications~\cite{DBLP:conf/ppopp/PerrinMJ16}. Most consistency
frameworks also do not allow describing abstractions involving (conceptual)
re-execution, such as OpSets~\cite{DBLP:journals/corr/abs-1805-04263} or
ECROs~\cite{DBLP:journals/pacmpl/PorreFPB21}.

\paragraph{Consistency axioms} Defines three basic consistency axioms
(\emph{monotonic visibility}, \emph{local visibility}, and \emph{closed
past}), which classic consistency models obey, to avoid being too weak or able to
produce too strange outcomes. The first two are generalizations of two of the
four classic session guarantees~\cite{DBLP:conf/pdis/TerryDPSTW94}, which are
more fundamental than the other two. The third allows distinguishing models in
which operations are final, from those that have (conceptual) re-execution of
operations.
Introduces \emph{Serial Consistency}, stronger than the combination of the
three basic axioms, being the bottom of the taxonomy of classic consistency
models, in which executions can be explained by serial executions, playing the
same role as \emph{Local Consistency} in the frameworks by
\textcite{DBLP:conf/europar/BatallerB97, DBLP:journals/jacm/SteinkeN04}, here
generalized beyond serial views.
Different combinations of a subset of the basic axioms allow
novel models to be obtained, with histories that cannot be explained by serial
executions. This allows a precise characterization of some existing
consistency models (such as \emph{Consistent
Prefix}~\cite{DBLP:journals/cacm/Terry13}), which have been up to now
described in an informal way, and allow them to be compared in a precise way.

\paragraph{Pipelined and Causal consistency}
Adequately defines Pipelined Consistency as a generalization of PRAM while
addressing a problem in the PRAM axiomatic
specification~\cite{DBLP:conf/spaa/AhamadBJKN93} which has been reused in many
places. That specification is missing a relevant axiom, being weaker than the
original operational definition~\cite{lipton1988pram} of PRAM. Importantly, it
allows reads resulting from a causality loop, not possible by the original
definition.
As far as we are aware, Pipelined Consistency as defined by our framework
(including well-formedness given by physical realizability) is the first
axiomatic definition which adequately generalizes PRAM to arbitrary data
types, while corresponding to the original operational definition of PRAM when
instantiated for memory. Our specification of Causal Consistency, when instantiated
for memory, also corrects a similar problem in the original specification of
Causal Memory~\cite{DBLP:journals/dc/AhamadNBKH95}, which does not satisfy the
stated intention/goal, namely being stronger than PRAM.  (It is stronger than
the axiomatic specification of PRAM, but not the original definition of PRAM,
also allowing physically impossible causality loops.)

Generalizes the definition of causal consistency, beyond sequential
specifications, as \textcite{DBLP:conf/ppopp/PerrinMJ16} generalized it beyond
memory, keeping it compatible with the intention of the original proposal for
memory, common understanding and current
taxonomies~\cite{DBLP:journals/jacm/SteinkeN04,DBLP:journals/csur/ViottiV16}.

Singles out a common base criterion, \emph{Causality}, which can be used to
obtain variants of causal consistency, by combining it with other criteria,
e.g., (classic) causal consistency, when combined with \emph{Serial
Consistency} and \emph{Causal Prefix Consistency} when combined with
\emph{Prefix Consistency}, also introduced.

\paragraph{Convergence as a safety property} Introduces convergence and
arbitration as safety properties for consistency models, through the
definition of \emph{Convergence} and \emph{Arbitration} in the
framework, allowing them to be combined with other constraints in an
orthogonal way. Current taxonomies~\cite{DBLP:journals/csur/ViottiV16} lack
these properties, including only \emph{Eventual Consistency}, which conflates
safety and liveness, something not useful for finite histories.
As an example, it is now possible to state that both causal consistency and
consistent prefix satisfy monotonic visibility and closed past, while the
former satisfies local visibility but not arbitration, and the latter
satisfies arbitration but not local visibility.


\paragraph{New trilemma for wait-free systems}
Formulates and proves the CLAM theorem for highly available partition tolerant
distributed systems (AP systems, from the
CAP theorem~\cite{DBLP:conf/hotos/FoxB99,DBLP:journals/sigact/GilbertL02}).
Essentially, a system which remains available even under partitions, i.e.,
with a wait-free implementation, cannot provide simultaneously
\emph{Closed past}, \emph{Local visibility}, \emph{Arbitration} and
\emph{Monotonic visibility}. As ``M'' is too fundamental, this results in
a new trilemma (CAL) about which of the other three properties to forgo.
Provides examples of how apparently disparate consistency models
or systems, such as Strong Update Consistency~\cite{DBLP:conf/ipps/PerrinMJ15}
or Consistent Prefix~\cite{DBLP:journals/cacm/Terry13} can be compared by how
they position under this trilemma, by forgoing one of these properties. The
trilemma clarifies that wait-free systems providing \emph{Convergent Causal
Consistency}, i.e., causal consistency with convergence, are not ruled out.
The dilemma of not being able to simultaneously achieve causal
consistency and causal convergence for wait-free
systems~\cite{DBLP:conf/ppopp/PerrinMJ16} does not apply to concurrent
specifications, nor to all sequential data types, only to those which rely on
arbitration for convergence.

\paragraph{Taxonomy} Presents a taxonomy of the main timeless consistency
models, using the new framework, according to the different criteria for
visibility and serializations. The taxonomy illustrates how the new framework
and newly introduced axioms allow very different models to be compared,
as different combination of basic axioms, something not possible under
current frameworks. With the exception of some use of session guarantees,
current taxonomies simply place most models above ``weak'' consistency, and do
not include models not satisfying serial consistency.

\subsection{Notation}

A binary relation is a set of pairs. For relation $\rel R$, we use $a \rel R b$ for
$(a, b) \in {\rel R}$. Given relation $\rel R$ and set $S$, we use
$\rel R y \defeq \{ x | x \rel R y\}$ for the set of 
values related to $y$ and
${\rel R}\domainrestrict{S} \defeq \{ x \rel R y | x \in S \land y \in S\}$
for the restriction of relation $\rel R$ to set $S$.
We use $\rel R^+$ for the transitive closure and $\rel R^=$ for
the reflexive closure of a relation $\rel R$.
We use $\max$ for the maximals of a partially ordered set.
We use standard notation for set or list comprehensions, sometimes together
with pattern matching. e.g., $\{P[x] \in S\}$ as shorthand for
$\{x | x \in S \land x \text{ matches } P\}$.
We may also use patterns or relations in quantifications: $\forall P[x] \rel R y$
as shorthand for
$\forall x \cdot x \text{ matches } P \land x \rel R y$.

\section{The Framework}

\subsection{Events, histories and executions}

A local history of a process is the sequence of events, corresponding to
the operations executed sequentially by the process. There is no notion of time,
or operations taking time, or the need to distinguish invocation and response
events: an event corresponds to the full execution of an operation, possibly
returning a result. The only assumption is that an operation returns the
result before the next operation of the same process starts being executed.
A distributed history is a collection of local histories, with no information
about the relative order of operations from different processes.

In Section~\ref{sec:time} we show how information about time, either physical or
logical, implementation-level events such as message sending, and start/finish
events of an operation can be abstracted away from consistency models, and
used as an orthogonal ingredient for optional time constraints, while keeping
histories and consistency models timeless.

A process is the abstraction for a context of sequential execution of a
series of operations, and can correspond to a process, a thread, a node or
replica, or a session of a client of a distributed data store service,
possibly replicated in several nodes. Implementation details are irrelevant
for the model. The model is about guarantees regarding what the process, as
client of the shared abstraction, can expect.

Events are written as, e.g., $\af{deq}^o_i(q){:}5$. In this case an event
we can refer to by unique identifier $o$, performed by process $i$, was a dequeue
on some object $q$, and returned 5. We use $\pr(o)$, $\op(o)$, and $\re(o)$, for
the process, operation and result of event $o$, respectively, in this example
$i$, $\af{deq}(q)$, and $5$. Process, identifier, or result of an operation
are hidden when irrelevant or understood from context. We use $o_i$ to
denote an event performed by process $i$.

Objects are mostly used for exemplification purposes, as they are not relevant
for most consistency models, in particular those discussed here.
When operations are defined over objects or locations, consistency models
where these play a role, such as Slow, Cache, or Processor consistency, can be
easily described in this framework, using $\obj(o)$ in constraints to
denote the object accessed by the operation. We refrain to present such models
with per-object semantics, to avoid cluttering the taxonomy, as no significant
insight would be gained. They form a partial order of models between
what we define below as \emph{Serial Consistency} and the classic
\emph{Sequential Consistency}.

A distributed history $(H, \po)$ is a partially ordered set of events $H$,
with events from the same process being totally ordered according with the
respective process sequence, by the so called \emph{program order} $\po$.

\begin{definition}[Program order]
Events $a$ and $b$ are related by program order, written $a \po b$,
if and only if they are from the same process, with $a$ before $b$ in the
local history.
\end{definition}

Depending on the shared abstraction and consistency model, the set of possible
histories are those which can be explained by the existence of a
\emph{visibility} relation and a \emph{serialization} relation per process, that:
\begin{itemize}
  \item are well-formed as described below;
  \item explain the result of each operation, according with the data type semantics;
  \item comply with some constraints, depending on the consistency model.
\end{itemize}
Towards this we define an abstract execution, or simply, execution.

\begin{definition}[Execution]
An execution $(H, \po, \vis, \{\ser_i\}_{i \in \ids})$, for a given history $H$
with program order $\po$, is the enhancement of ($H, \po$) with a
visibility relation $\vis$, and a set of serialization relations $\ser_i$, for
each process $i$ in the set of processes $\ids$.
\end{definition}

\subsection{Visibility}

Visibility is an irreflexive binary relation describing the possibility of effects
between events.

\begin{definition}[Visibility]
Operation event $a$ is visible to $b$, written $a \vis b$, if the effect of
$a$ is visible to the process performing $b$, when performing $b$.
\end{definition}

This means that $a$ can be considered in the computation of the outcome of
$b$, if relevant, according with the data type semantics.
If $a$ is not visible to $b$, from the point of view of $b$, it is as if $a$
has not been performed, even if it has occurred much in the past, even in the
same process (something undesirable, which does not occur in classic
consistency models).

We do not restrict operations to act on single objects (which would not allow,
e.g., an operation which performs the dot product of two vectors), nor we
define visibility only between operations on the same object, or even
operations on the same data type. An operation $a$ being visible to $b$ just
means that it is available to being accessed by the semantic function for $b$
to derive its result, if it is a query. The semantic function can simply
ignore operations on other objects, or operations of other data types. As the
main consistency models are not per-object based, only for those which are
does object-based notation need to be introduced, being otherwise a
distraction.
Visibility as defined here mimics information propagation between
processes, relevant to the abstraction, regardless of the information being
relevant to specific operations.

We also do not need to distinguish updates and queries. Pure updates have no
result, and can be considered to return the $()$ value of the unit type, with
a semantic function which ignores visible operations. A pure query has no
effect, but it can still be considered visible to other operations (and
ignored by their semantic functions) to meet constraints for visibility that
can be kept simple and elegant.
Two examples are requiring that visibility includes the program order for each
process, and that visibility is a total order of all events in the history.
More restrictive definitions of visibility (e.g.,
from~\cite{DBLP:conf/popl/BurckhardtGYZ14}) lead to more clutter and less
elegant formulations for consistency models.

To define a condition on what possible visibility relations can be assumed,
for a given history, it is useful to define semantic happens-before, or simply
\emph{happens-before}, between operations:

\begin{definition}[Happens-before]
  Operation event $a$ happens-before $b$, written $a \hb b$, if and
  only if: 1) $a \po b$, or 2) $a \vis b$, or 3) there is some event $c$ such
  that $a \hb c$ and $c \hb b$.
\end{definition}

This can be more succinctly expressed using the transitive closure:
\[
{\hb} \defeq (\po \union \vis)^+
\]
and it is analogous to the classic \emph{happened before}
by~\textcite{DBLP:journals/cacm/Lamport78} defined for messages. In
Section~\ref{sec:time} we discuss several variants of happens-before.

\begin{definition}[Physical realizability axiom]
For some distributed history $(H, \po)$, a visibility
relation $\vis$ satisfies physical realizability if and only if:
  \[
    a \hb b \implies b \not\po a
  \tag{$\axiom{vis}{pr}$}
  \]
\end{definition}

The reason is that visibility must be implemented by some information
propagation mechanism, such as message passing, limited by the physical
happens-before, given by light-cones (which we detail below). Effects of an
operation $a$, or of other operations that see the effects of $a$, or of any
other operation that is subsequent or sees the effects of any such operation
cannot be observed by an operation $b$ that has computed its result before $a$
started executing.  Doing so would create a cycle bringing information from
the future.

Some examples are given in Figure~\ref{fig:vis-pr}, for two processes. In cases
(a) and (b), the visibility does not satisfy physical realizability, as
$b$ occurs before $a$ in the same process, while being transitively reachable
from $a$ by following program order or visibility edges.
The case in Figure~\ref{fig:vis-pr-c} satisfies physical realizability. There is
a cycle but only involving visibility between operations from different
processes, not program order.
In defining the physical realizability axiom we consider two assumptions about
operations and their semantics in our framework:
\begin{itemize}
  \item An operation is the basic unit of computation, having atomic
    visibility of its effects. (The A in ACID, as all-or-nothing, but nothing
    else) An operation $a$ cannot be some long running computation which is
    producing some effects for operation $b$ and some different effects for
    operation $c$. The implementation must somehow ensure it. If it is a
    low level abstraction such as memory, the implementation must still ensure
    that a write to a location is either seen or not as a unit, preventing word-tearing
    if a location involves to two machine words. If it is a higher-level
    abstraction, such as a stack with a ``double push'', it must ensure that
    either two elements are seen pushed or none.
  \item If an operation $a$ is visible to $b$, then the effect of $b$
    \emph{may} depend on it, being different if $a$ was not visible.
    Therefore we cannot let $b$ be visible to an operation $c$ which has
    completed before $a$ started.
\end{itemize}

\begin{figure*}
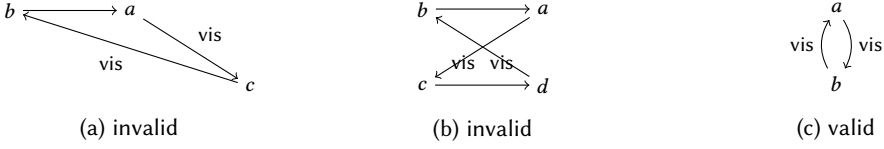

  \newlength\h
  \setlength\h{16mm}
  \footnotesize
  \begin{boxes}{3}
    \begin{tcolorbox}
      \tikz \graph [grow right=\h] {
        b/$b$ -> a/$a$;
        c/$c$ [x=2\h];
        a ->["$\vis$"] c;
        c ->["$\vis$"] b;
};
      \subcaption{invalid}
      \label{fig:vis-pr-a}
    \end{tcolorbox}
    \begin{tcolorbox}
      \tikz \graph [grow right=\h] {
        b/$b$ -> a/$a$;
        c/$c$ -> d/$d$;
        a ->["$\vis$"] c;
        d ->["$\vis$"] b;
};
      \subcaption{invalid}
      \label{fig:vis-pr-b}
    \end{tcolorbox}
    \begin{tcolorbox}
      \tikz \graph [grow right=\h] {
        a/$a$;
        b/$b$;
        a ->["$\vis$", bend left] b;
        b ->["$\vis$", bend left] a;
};
      \subcaption{valid}
      \label{fig:vis-pr-c}
    \end{tcolorbox}
  \end{boxes}
  \caption{Operations from each process from left to right. Well-formedness of
  visibility relations according to the physical realizability axiom.}
  \label{fig:vis-pr}
\end{figure*}

Although the model is timeless, with histories containing no time
information this axiom uses the knowledge that if $a$ precedes $b$ in some
process, the completion of $a$ precedes the start of $b$ in physical time.
Histories for linearizability or in some frameworks with time
information~\cite{DBLP:journals/ftpl/Burckhardt14} can define a global
returns-before relating operations from different processes, but we are not aware
of any timeless framework having the physical realizability axiom.
As we will see, axiomatic models weaker than sequential consistency lacking this
axiom will allow physically impossible histories.

This axiom allows happens-before not to be a strict partial
order but a preorder.  A more restrictive axiom, sometimes used for models
involving happens-before, such as causal consistency, prevents any
cycle involving program order or visibility:
\[
  \po \union \vis \text{ is acyclic}
  \tag{$\axiom{vis}{no-cycles}$}
\]
This alternative axiom makes happens-before between operations a strict partial
order. It can be used to constrain visibility in almost all cases.
Histories resulting from shared abstractions with wait-free implementations can
be explained while satisfying it. It can even be used by abstractions with
partial sequential specifications whose implementations resort to blocking.
An example is a queue with a partial sequential
specification~\cite[p.472]{DBLP:journals/toplas/HerlihyW90},
where a dequeue is undefined for an empty queue, and where the implementation
blocks the dequeue operation until an enqueue is invoked.

But preventing all cycles, as usually done, is overly restrictive according to what is
physically possible, and does not allow histories resulting from
synchronization-oriented abstractions with blocking implementations, such
as barrier-like abstractions. Two examples are presented in
Section.~\ref{sec:semantics}, a Barrier and a Consensus abstraction, both
defined resorting to partial concurrent specifications. The physical realizability
axiom allows such abstractions.
Preventing acyclic happens-before through \axiom{vis}{no-cycles} is,
therefore, an optional axiom which we may want some implementations to
satisfy, but not compulsory as physical realizability.

\subsection{Serializations}

If data types only with commutative update operations are involved, visibility is
enough, and there is no need for serializations. This is also the case for
CRDTs~\cite{DBLP:conf/sss/ShapiroPBZ11}, that can resolve conflicts without the
need for arbitration, as discussed below. But in general, when sequential
specifications are used, total orders are required.

We use the term \emph{serialization} as used, e.g., by
\textcite{DBLP:conf/spaa/AhamadBJKN93,DBLP:journals/dc/AhamadNBKH95},
for totally ordering a set of events into a sequence, possibly is some
``strange'' way. This term is more general than
\emph{linearization}~\cite{DBLP:journals/toplas/HerlihyW90}, which refers to a
serialization which respects a precedence partial-order according with
``real-time'' ($a$ precedes $b$ if it is before $b$, not overlapping $b$ in real
time). Not only our framework for consistency models is timeless (with time
constraints as an optional orthogonal ingredient), without operations taking
time or overlapping in time, but we need to allow serializations that
contradict program order.

The purpose of a serialization is to be used in the semantic function, for
sequential specifications, which determines the result of query operations. It
defines the order in which effects of visible operations are applied, starting
from the initial state, to obtain the current state, observable by some query
operation.
If operation $a$ is visible to operation $b_i$, at process $i$,
then its effect has taken place at $i$ when $b_i$ executes; therefore, $a$
must be serialized before $b_i$ in $\ser_i$.

\begin{definition}[Serialization of visibility axiom]
For any execution $(H, \po, \vis, \{\ser_i\}_{i \in \ids})$, serialization
$\ser_i$ for process $i$ is a strict total order on the events in $H$ that must
respect, for any event $a$ and event $b_i$ by process $i$:
\[
  a \vis b_i \implies a \ser_i b_i
  \tag{$\axiom{ser}{vis}$}
\]
\end{definition}

Including all events in the history in each serialization allows elegant
formulations (e.g., of sequential consistency or arbitration, see below).
Serializations reflect the program order of the operations issued by
the process itself, unless the basic \emph{local visibility}, described below,
is not met.
Visible operations from another process may be serialized
not in program order, depending on the consistency model.
This is the case when pipelined consistency is not satisfied (see below).


\subsection{Well-formed executions}

An execution which satisfies both the physical realizability and 
serialization of visibility axioms is said to be well-formed. Being
well-formed is a pre-requisite to submitting an execution to be checked as
valid, given a data abstraction semantics, as we describe below.

\begin{definition}[Well-formed execution]
An execution $(H, \po, \vis, \{\ser_i\}_{i \in \ids})$ is well-formed if and
only if it satisfies:
  \[ 
  \axiom{vis}{pr} \land
  \axiom{ser}{vis}
  \tag{$\axiom{exe}{wf}$}
  \]
\end{definition}

\subsection{Data abstraction semantics}
\label{sec:semantics}

Each data abstraction defines a semantic function $\sfR{}$ for what result an
operation must return in a given execution. It may use the process
serialization to order visible operations (sequential specifications) or just
the visibility (concurrent specifications, e.g., in CRDTs).

\begin{definition}[Result validity axiom]
An execution $(H, \po, \vis, \{\ser_i\}_{i \in \ids})$, satisfies result
validity if and only if, for each event $o_i$, by any process $i$:
\[
  \re(o_i) = \sfR{o_i,\vis,\ser_i}
  \tag{$\axiom{res}{val}$}
\]
\end{definition}

Two classes of specifications are relevant. The first is for standard
sequential data types, through a semantic function $\sfS{}$ that takes the
operation and a sequence of the previous operations, and applies them in
sequence, starting from the initial state. For these, serialization defines
the order over the visible operations:
\[
  \sfR{o_i,\vis,\ser_i} = \sfS{\op(o_i),
  [\op(x) | x \in {\ser_i}\domainrestrict{\vis\, o_i}]}
\]
(For notational convenience we use $\ser_i$ either as a relation between
events or a list of events, to write list comprehensions, as here.)
The possibility of restricting the operations in a serialization to the
visible subset is the main novel feature of our framework, when compared
with the more standard, that only uses serializations, with the implicit
assumption that every operation serialized before a given one is visible to
that operation (has taken effect before). This feature allows our framework to
be more expressive, being able to specify systems that (conceptually)
re-execute operations.

However, as we discuss below, such systems allow somewhat strange outcomes,
which are not explainable as serial executions consistent with the program,
for each process, even allowing them to be different for different processes.
Classic consistency models do not apply this visibility based restriction.
When described in this framework, they satisfy an axiom which restricts
admissible serializations, \emph{closed past}, which we introduce below.

The second class of specifications only use visibility, with no need for
serializations. They allow concurrent specifications, using a semantic
function $\sfC{}$ that can use visibility not only for the set of visible
events, but as defining the concurrency relation between them:
\[
  \sfR{o,\vis,\ser_i} = \sfC{o, \vis}
\]

These have started being used with CRDTs, that allow concurrent specifications
(such as the Observed-Remove Set or the Multi-Value Register), whose outcome
is not equivalent to any serialization. It can also be used for data types
containing only commutative operations (such as counters and grow-only sets),
with semantics equivalent to some permutation of the visible updates
(\emph{Principle of Permutation Equivalence}~\cite{DBLP:conf/wdag/BieniusaZPSBBD12}).

As we normally discuss single executions at a time, we write semantic
functions implicitly parameterized over $\vis$ and $\ser_i$, to remove
clutter.  Table~\ref{tab:concurrent-specs} presents specifications using only
visibility for some well known data types: counter, grow-only set,
observed-remove set (where a remove acts on visible adds, making an add
``win'' over concurrent removes of the same element), and the multi-value
register (where a read returns the set of the more recent concurrent
writes). These specifications show how all events, from all objects, of all
data types can be considered together in the history, and be potentially
visible, and how the semantic function that computes a result can pick only
the relevant ones, ignoring the others. E.g., when computing the value of a
counter $c$, only the increments to that counter, $\inc(c, v)$, for some $v$,
will be considered.

\begin{table}
  \caption{Some concurrent data type specifications based on visibility (and
  implicitly parameterized over it). Counter, GSet, ORSer and MVRegister are
  traditional in CRDTs, allowing wait-free implementations. Barrier and
  Consensus are synchronization objects requiring blocking by the
  implementations and cyclic visibility.}
  \label{tab:concurrent-specs}
  \auxfun{wait}
  \auxfun{epoch}
  \auxfun{propose}
  \begin{center}
    \begin{tabular}{@{}ll@{}}
  \toprule
    Data type & Semantic function \\
  \midrule
    \multirow2*{Counter} & $\sfC{\inc(c, v)} = ()$ \\
    & $\sfC{\val^o(c)} = \sum \{ v | \inc(c, v) \vis o \}$ \\
    \cmidrule{1-2}
    \multirow2*{GSet} & $\sfC{\add(s, v)} = ()$ \\
    & $\sfC{\val^o(s)} = \{ v | \add(s, v) \vis o \}$ \\
    \cmidrule{1-2}
    \multirow3*{ORSet} & $\sfC{\add(s, v)} = ()$ \\
     & $\sfC{\rmv(s, v)} = ()$ \\
     & $\sfC{\val^o(s)} = \{v | \exists \add^a(s,v) \vis o \cdot
       \not\exists \rmv^b(s, v) \cdot a \vis b \vis o \}$ \\
    \cmidrule{1-2}
    \multirow3*{MVRegister} & $\sfC{\af{wr}(s, v)} = ()$ \\
     &
       $\sfC{\af{rd}^o(r)} = \{ v | \exists \af{wr}^a(r,v) \vis o \cdot
       \not\exists \af{wr}^b(r, \_) \cdot a \vis b \vis o \}$ \\
    \cmidrule{1-2}
    \cmidrule{1-2}
    Barrier$[n]$ &
     $\begin{aligned}
       \sfC{\wait^o(b)} = {} & ()
       \text{ if } \setsize{S} = n \land \forall a\in S, b \po o \cdot a \not\vis b \\
       & \text{ with } S \defeq \{\wait(b) \in \max(\vis^= o)\}
     \end{aligned}$ \\
    \cmidrule{1-2}
    Consensus$[n, f]$ &
     $\begin{aligned}
       \sfC{\af{propose}^o(c, \_)} = {} & f(\{v | \propose(c, v) \in S\}) \\
       & \text{ if } \setsize{S} = n \land \forall a\in S, b \po o \cdot a \not\vis b \\
       & \text{ with } S \defeq \{\propose(b) \in \max(\vis^= o)\}
     \end{aligned}$ \\
  \bottomrule
  \end{tabular}
  \end{center}
\end{table}

The same figure also shows two partial concurrent specifications for two
synchronization-oriented data abstractions, that need blocking
implementations. These are
possible in our framework, explainable by a cyclic visibility. A wait on a
barrier $b$, parameterized by the number of processes $n$, can
return (with no result) when the set $S$ of waits in the most recent
concurrent operations, including itself (i.e., the waits in the
maximals under the reflexive closure of visibility) has $n$ elements, and all
those concurrent waits were not already visible to previous operations from
the process (i.e.,
``used'' for a previous wait). A consensus abstract data type is similar to a
barrier, but the operation $\af{propose}$ takes a value. It waits for $n$
concurrent proposals and returns the same value to each process (usually one
of the proposed values), computed by a function over the set of proposed values.

Table~\ref{tab:abstraction-examples} shows some classes of abstractions and
examples, according to: whether they have sequential or concurrent
specification, partial or total; if histories can be explained by acyclic
visibility or require visibility cycles; and the nature of implementation for
asynchronous distributed systems as either immediate, wait-free or blocking.
Immediate is a special case of wait-free, in which implementation uses local
state to immediately answer queries and then propagate effects asynchronously;
a wait-free implementation can
in general interact with other processes ``for some time'' and ``happen to see
their effects'' but does not depend on them to complete; a blocking
implementation depends on receiving messages to complete some operations, and
will not tolerate network partitions. (Note that some
abstractions that can allow wait-free implementations for shared-memory,
resorting to operations guaranteed to terminate in finite time, such as
Compare-and-Swap, will have to be blocking in asynchronous distributed
systems.)

\begin{table}
  \caption{Examples of shared abstractions, for different combinations regarding:
  sequential or concurrent specifications, whether total or partial;
  visibility (acyclic or cyclic);
  immediate, wait-free, or blocking implementation.}
  \label{tab:abstraction-examples}
\begin{tabular}{@{}lllll@{}}
  \toprule
  specification & visibility & implementation & examples \\
  \midrule
  sequential, total & acyclic & blocking & atomic objects \\
  sequential, partial & acyclic & blocking & bounded-buffer \\
  concurrent, total & acyclic & immediate & CRDTs \\
  concurrent, total & cyclic & wait-free & write\_snapshot~\cite{DBLP:journals/jacm/CastanedaRR18} \\
  concurrent, partial & cyclic & blocking & barrier, consensus \\
  \bottomrule
\end{tabular}
\end{table}

We will also present examples using other well known data types with
sequential specifications, like queues or stacks, whose specifications we omit.

\subsection{Valid executions}

A valid execution is one which is well-formed (the visibility and
serializations satisfy $\axiom{vis}{pr}$ and $\axiom{ser}{vis}$),
and when the result of each operation is valid according with the
specification.

\begin{definition}[Valid execution]
An execution $(H, \po, \vis, \{\ser_i\}_{i \in \ids})$ is valid if and only if
it satisfies:
  \[ 
  \axiom{exe}{wf} \land
  \axiom{res}{val}
  \tag{$\axiom{exe}{val}$}
  \]
\end{definition}

A valid execution satisfies a consistency model if the relations satisfy the
model constraints.  The differences between consistency models are which
constraints the visibility and serializations must satisfy. The stronger the
consistency model, the more constrained those relations.
For the remainder of the article we only consider valid executions,
unless otherwise stated.

\section{Time and consistency models}
\label{sec:time}

Most classic consistency models (Cache, Processor, Pipelined, Causal, and Sequential
Consistency) do not refer to time, and can be specified in this framework.
The main well known exception is
Linearizability~\cite{DBLP:journals/toplas/HerlihyW90}, a time-based model,
which has as input a totally ordered history, according to points in a global
background of ``real-time'', that define when each operation starts and
finishes, to assess overlap or precedence, and restrict acceptable histories.

Making histories involve a total order according to ``real-time'' is
problematic in two ways. First it causes fragility in gathering the input to
be checked for a given consistency model, if we want to check actual
distributed runs and not merely reason about made up histories as an omniscient
observer. Consider an example with a totally
ordered history with two operations, where process $i$ reads $42$ and
later process $j$ writes $42$. If submitting this history to be checked
for some consistency model, it is immediately rejected as impossible to have
happened, unless $i$ could ``predict the future''. A more mundane explanation
is that the input was gathered at different nodes with slightly unsynchronized
clocks, and the write was indeed before the read. The difficulty, and
fragility, of gathering a totally ordered history according to ``real-time''
is a good reason to keep time out of consistency models, and why most do so.
The second reason is that there is no such thing as a ``totally ordered
background of real-time''. Events in spacetime~\cite{Minkowski-spacetime} are
partially ordered according to light cones.

Consider the spacetime diagram in Figure~\ref{fig:light-cone}. Events in the
future light cone of $o$, such as $d$, can be influenced by $o$; events in the
past light cone of $o$, such as $a$ can influence $o$. On the other hand, $b$ and
$c$ can neither influence or be influenced by $o$; they form \emph{spacelike}
intervals and comparing their times is physically meaningless: in this diagram
one could think that $b$ is before $o$ in time, but for a different observer
$o$ can be before $b$. (\textcite{SalgadoMAA2016} allows experimenting with
spacetime diagrams to gain insight.)
Only the comparison of events forming a \emph{timelike} interval such as $a$
and $o$, or $o$ and $d$ remains invariant for all observers and is physically
meaningful. We have, therefore a partial order:
event $o$ can potentially influence $d$, and so \emph{physically
happens-before} $d$; events $b$ and $o$ cannot influence each other, being
unrelated in time.

  \begin{figure}
  \begin{center}
  \includegraphics[scale=0.6]{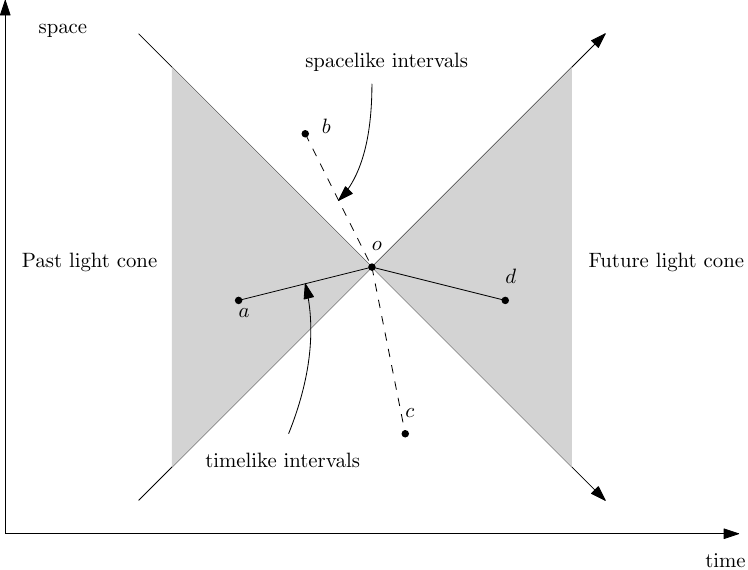}
  \end{center}
    \caption{Past and future light cones of event $o$, space-like and
    time-like intervals.}
    \label{fig:light-cone}
  \end{figure}

\begin{definition}[Physical happens-before]
Event $a$ physically happens-before $b$, written $a \phb b$, if and only if $b$ is
contained within the future light cone of $a$.
\end{definition}

So, $a \phb b$, if and only if it is possible that information about $a$
reaches $b$, i.e., $a$ can influence $b$.
A given distributed system, using some message passing, will exploit a subset of
what is physically possible. This is expressed by the classic \emph{happened
before} by~\textcite{DBLP:journals/cacm/Lamport78}, which we call here
\emph{messaging happens-before}.

\begin{definition}[Messaging happens-before]
  Event $a$ messaging happens-before $b$, written $a \mhb b$, if and
  only if: 1) $a \po b$, or 2) $a$ is a send and $b$ a receive of some
  message, or 3) there is some event $c$ such that $a \mhb c$ and $c \mhb b$.
\end{definition}

Messaging happens-before is a strict partial order contained in the physical
happens-before, i.e.,
\[
  {\mhb} \subset {\phb},
\]
and it involves implementation-level events (send, receive), which are hidden
from the consistency model of a shared abstraction.

Messaging happens-before is about potential causality, given a communication
pattern, but not about the communication that is relevant for the
shared abstraction being implemented. The relevant potential causality,
according with data type semantics, is given by what we have called
simply happens-before, defined similarly, but using visibility instead of
send/receive.

As we discuss in the next section, an operation ``takes some time'' during
which information can be sent or received, involving several messages.
Consider for now the special case of abstractions that allow wait-free
implementations with ``instantaneous'' operations, that use local knowledge and
trigger a single message (possibly broadcast).
Because, for efficiency reasons, implementations normally do not send the full
knowledge (transitive causal past) in messages, happens-before is a subset of
messaging happens-before, and so:
\[
  {\hb} \subseteq {\mhb} \subset {\phb}.
\]

To compare these concepts, consider the well-known example of \emph{causal
broadcast}, using a classic algorithm~\cite{DBLP:journals/tocs/BirmanSS91} that
buffers messages until causal predecessors arrive. The abstraction events
are $\af{cbcast}$ and $\af{deliver}$.
In terms of consistency model, if $a$ is a $\cbcast$ and $b$ a respective
$\deliver$, then $a \vis b$. Causal broadcast aims to ensure that:
\[
  \cbcast_i(a) \hb \cbcast_j(b) \implies \forall k \cdot \deliver_k(a) \po \deliver_k(b)
\]

The implementation events (hidden from the abstraction and consistency model)
are $\af{send}$ and $\af{receive}$. Consider Figure~\ref{fig:cbcast}, where
these four events are abbreviated as ``c'', ``d'', ``s'' and ``r'', respectively.

\begin{figure}
\begin{center}
\center\includegraphics[scale=0.9]{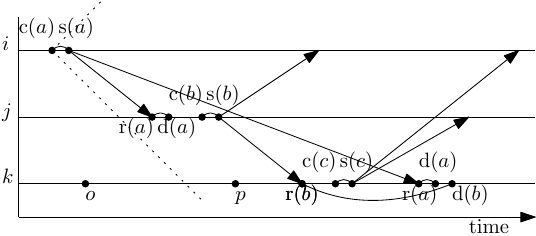}
\end{center}
  \caption{Causal broadcast implemented by buffering messages with missing
  dependencies; abstraction events: \textbf{c}bcast, \textbf{d}eliver;
  implementation level events: \textbf{s}end and \textbf{r}eceive.}
  \label{fig:cbcast}
\end{figure}

Upon a cbcast, the respective message is sent to all processes (self-sends and
delivers are omitted here). When a message is received, it may need to be
buffered for some time before delivery. In this example, when $b$ is received
by process $k$, it cannot be delivered because $a$ is missing; it is buffered,
until $a$ is received, upon which $a$ is delivered and then $b$ is delivered.

Event $o$ in the figure is not inside the light cone of event $\cbcast(a)$,
depicted by dotted lines, therefore $\cbcast(a) \cnot \phb o$. That $o$ is
slightly after in time in this spacetime diagram is irrelevant (it could be
slightly before). On the other hand, $\cbcast(a) \phb p$, but because the
possibility of propagating information was not used (no message), then
$\cbcast(a) \cnot\mhb p$.

We have $\cbcast(a) \hb \cbcast(b)$ as $a$ has already been delivered:
$\cbcast(a) \vis \deliver(a) \po \cbcast(b)$.
We have $\cbcast(a) \mhb \cbcast(c)$ and also $\cbcast(b) \mhb \cbcast(c)$.
Therefore, with this messaging pattern it could be possible that some other
algorithm (with message payloads carrying all causal predecessors) would allow
earlier delivery of both $a$ and $b$, and make $\cbcast(a)$ visible to $\cbcast(c)$.
But this algorithm sends individual messages and is still
waiting for $a$ to arrive. Therefore, no effect (either about $a$ or
$b$) has become visible yet, and 
$\cbcast(a) \not\hb \cbcast(c)$ and also $\cbcast(b) \not\hb \cbcast(c)$.
Message $b$ has already arrived but is being buffered and has not yet been
delivered. So, $\cbcast(c)$ is semantically concurrent to the other two. This
is what matters for the abstraction, and consistency model.

This example illustrates that the presence of physical or messaging happens-before
cannot imply the occurrence of visibility and semantic happens-before. What
time/messaging can imply for the consistency model is the \emph{absence}
of visibility, i.e., what events cannot be visible to others.
If we have information about physical time, we can use the implication:
\[
  a \cnot\phb b \implies a \cnot \vis b
\]
If we have information about messaging, we can use the implication:
\[
  a \cnot\mhb b \implies a \cnot \vis b
\]
which further constrains possible visibility, because it describes actual
messaging that occurred, versus physically possible messaging.

\subsection{Time-constraining valid executions}

What we propose, as a robust way to use information about time, while keeping
consistency models themselves timeless, is to submit valid executions
(satisfying $\axiom{exe}{val}$) to be
checked against time-related constraints to obtain \emph{time valid
executions}. These are checked against consistency models. As opposed to using
time just for some consistency models, this step is optional and orthogonal,
and allows several variants of each consistency model (unconstrained,
messaging constrained or physical time constrained). This is illustrated in
Figure~\ref{fig:executions}.

\begin{figure}
\begin{center}
\includegraphics[scale=0.8]{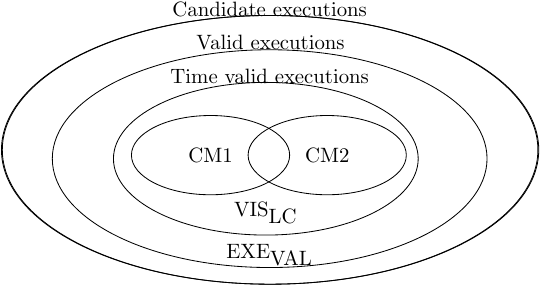}
\end{center}
  \caption{Candidate executions are checked: first for validity
  ($\axiom{exe}{val}$); then, optionally, for time constraints (e.g.,
  $\axiom{vis}{lc}$); finally, against specific consistency models.}
  \label{fig:executions}
\end{figure}

Because access to physical time is difficult and fragile, a robust way
to incorporate time constraints is if we have access to a \emph{logical
clock}~\cite{DBLP:journals/cacm/Lamport78}. Lamport scalar clocks in
particular are very lightweight and can be built by piggybacking each message
with a single integer. A logical clock $L$ satisfies the clock condition
\[
  a \mhb b \implies L(a) < L(b),
\]
where $<$ may be a total (for scalar clocks) or a partial order.
With a logical clock we have that:
\[
  L(a) \cnot< L(b) \implies a \cnot\mhb b \implies a \cnot\vis b.
\]

But while in the abstract model an event is timeless, its implementation may
take some ``time'' and involve some message passing until it completes. This
passage of time, for an operation $o$, can be abstracted by considering the
logical clock value when the operation starts, which we write as
$L(o)$, and when it completes, $L'(o)$.
As illustrated by Figure~\ref{fig:possible-visibility}, if the clock value at
start of $a$ is not less than the clock value upon completion of $b$, then $a$
cannot be visible to $b$.

\begin{figure}
\begin{center}
\includegraphics[scale=0.8]{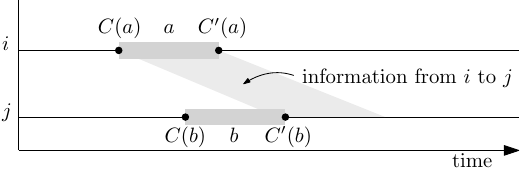}
\end{center}
  \caption{Operation $a$ can be visible to $b$ only if clock $C(a)$ at
  start of $a$ is less than clock $C'(b)$ at completion of $b$, whether for
  logical or physical clocks, either partially or totally ordered.}
  \label{fig:possible-visibility}
\end{figure}

\begin{definition}[Logical clock axiom]
A visibility relation $\vis$ satisfies the logical-clock axiom if and only if,
for any events $a$ and $b$:
\[
a \vis b \implies L(a) < L'(b)
  \tag{$\axiom{vis}{lc}$}
\]
\end{definition}

This allows gathering histories from actual runs of some distributed system,
tagged with logical clock values, which allows, optionally, restricting the set
of candidate executions that are existentially assumed as explaining those
histories while satisfying some consistency model. I.e., not only working in
an ``omniscient mode'', considering hypothetical histories, but enabling,
e.g., submitting actual runs to a consistency checker while taking some
information about time in consideration, in a robust way.

A similar reasoning applies to physical clocks, either partially ordered in
spacetime, or totally ordered (assuming some background of time).
If we do not have access to a logical clock, but consider histories tagged by
an omniscient observer, with access to such a physical clock $P$, we can
constrain valid executions in a similar way.

\begin{definition}[Physical clock axiom]
A visibility relation $\vis$ satisfies the physical-clock axiom if and only if,
for any events $a$ and $b$:
\[
a \vis b \implies P(a) < P'(b)
  \tag{$\axiom{vis}{pc}$}
\]
\end{definition}

This axiom clarifies the possibility of cyclic visibility allowed by the
physical realizability axiom, where operations mutually influence each
other, when the execution of each operation overlaps the light cone at the
start of the other operation.

As an example of applying this axiom, if a totally ordered physical clock is
assumed, the classic linearizability is obtained as a combination of this
axiom and sequential consistency. For the remainder of the article we focus on
the timeless framework.

\section{Basic consistency axioms and Serial Consistently}

We now present three basic axioms which we consider the more fundamental
building blocks for consistency models.  These are satisfied by classic
consistency models (Cache, Processor, Pipelined, Causal, and Sequential
Consistency), when specified in this framework. The first two, \emph{monotonic
visibility} and \emph{local visibility}, correspond roughly to the well known
\emph{monotonic reads} and \emph{read-your-writes}, two of the four session
guarantees~\cite{DBLP:conf/pdis/TerryDPSTW94}. (The session guarantees are,
however, defined resorting to implementation aspects: that processes are
clients to servers which replicate data, specifying which writes should have
reached servers and in which order.) The third, \emph{closed past}, is a
novelty, which allows defining models where operations are final, free from
re-execution, in contrast with those where operations need to be
(conceptually) reapplied until the past stabilizes.  We then introduce
\emph{Serial Consistency}, which ensures serializations which are
self-consistent as serial executions. It implies the three basic axioms, and
is the basic constraint satisfied by classic consistency models.

\subsection{Monotonic visibility}

The first condition is about whether the effects that become visible at each
process remain visible forever. Monotonic visibility expresses that the effect
of any operation (issued by any process), once seen by a process, is not
``unseen''. This makes the set of operations visible to each process increasing
along time.

\begin{definition}[Monotonic visibility]
An execution $(H, \po, \vis, \{\ser_i\}_{i \in \ids})$, satisfies
monotonic visibility if and only if, for any events $a$, $b$, $c$ in $H$:
\[
  a \vis b \po c \implies a \vis c \\
  \tag{$\axiom{vis}{mon}$}
\]
\end{definition}

This does not restrict propagation order, e.g., not respecting pipelined
consistency, only that there is a monotonic accrual of information at each
process. Any effect can only be overridden or complemented by other effects
becoming visible.

Classic models tend to resort to local serializations, and implicitly,
to a monotonic visibility, where writes remain until overwritten,
but we can argue for monotonic visibility even when no serialization is involved
in the data abstraction semantics, such as when only commutative
operations are involved or in CRDTs that converge while resolving conflicts
semantically, without the need for a global arbitration.

Monotonic visibility is arguably the most fundamental condition. Without it,
reasoning about progress in a process becomes hopeless. Not having it results
into something that does not correspond to the notion of process, but to
something that could be called ``byzantinely amnesic'', which can forget
arbitrary ``state'' and return different results to queries along time, after
reaching quiescence, when no updates are being issued, either locally or at
other ``processes''. Monotonic visibility holds not only in all classic
consistency models, but also in recently introduced models which forgo one
of the two basic conditions that we introduce next: \emph{local visibility} or
\emph{closed past}. It can be arguably deemed to be mandatory to any
``reasonable'' consistency model.

\subsection{Local visibility}

The second condition states the standard expectation from the point of view of
a single process: that its operations have an immediate effect upon the local
state, and that effect remains visible (it can only be overridden or
complemented by subsequent operations).

\begin{definition}[Local visibility]
An execution $(H, \po, \vis, \{\ser_i\}_{i \in \ids})$, satisfies
local visibility if and only if, for any events $a$, $b$ in $H$:
\[
  a \po b \implies a \vis b \\
  \tag{$\axiom{vis}{loc}$}
\]
\end{definition}

This condition is natural and implicitly assumed in shared memory models for
multiprocessors. In distributed systems, it needs to be made explicit, and may need
effort to be satisfied.
It is natural when a process is a replica, but not when a process is a
client to some replicated data store, in which case it needs enforcement. Some
ways to do it are restricting client affinity to the same server, or using
some mechanism to track process context when changing server (which may cause
delay and is not partition tolerant).

We note that local visibility only ensures that local operations remain
visible, but not the operations from other processes.
Conversely, monotonic visibility can hold without local visibility, as operations from
the process itself do not need to became immediately visible, only that once they
become visible they remain so.
Figure~\ref{fig:visibility-examples} exemplifies the four combinations of
local and monotonic visibility being or not being satisfied, considering a
single counter object.

\begin{figure*}
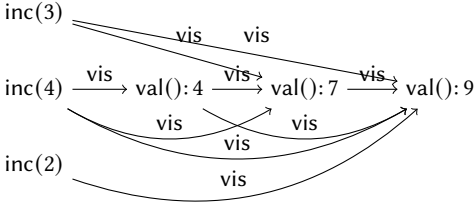
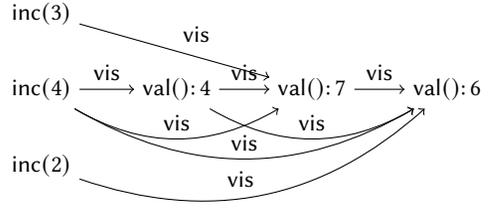
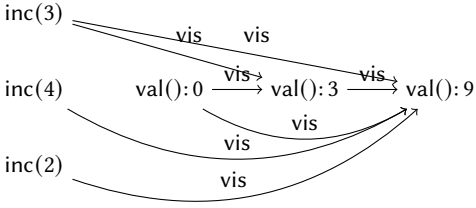
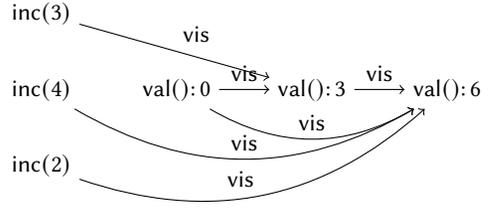

  \small
  \begin{boxes}{2}
    \begin{tcolorbox}
      \tikz \graph [edge label=${\vis}$, grow right=18mm] {
        a/$\inc(3)$;
        b/$\inc(4)$ -> c/$\val():4$ -> d/$\val():7$ -> e/$\val():9$;
        f/$\inc(2)$;
        a -> {d, e};
        b ->[bend right] {d, e};
        c ->[bend right] e;
        f ->[bend right] e;
};
      \subcaption{local and monotonic visibility}
    \end{tcolorbox}
    \begin{tcolorbox}
      \tikz \graph [edge label=${\vis}$, grow right=18mm] {
        a/$\inc(3)$;
        b/$\inc(4)$ -> c/$\val():4$ -> d/$\val():7$ -> e/$\val():6$;
        f/$\inc(2)$;
        a -> d;
        b ->[bend right] {d, e};
        c ->[bend right] e;
        f ->[bend right] e;
};
      \subcaption{local but not monotonic visibility}
    \end{tcolorbox}
    \begin{tcolorbox}
      \tikz \graph [edge label=${\vis}$, grow right=18mm] {
        a/$\inc(3)$;
        b/$\inc(4)$ -!- c/$\val():0$ -> d/$\val():3$ -> e/$\val():9$;
        f/$\inc(2)$;
        a -> {d, e};
        b ->[bend right] e;
        c ->[bend right] e;
        f ->[bend right] e;
};
      \subcaption{monotonic but not local visibility}
    \end{tcolorbox}
    \begin{tcolorbox}
      \tikz \graph [edge label=${\vis}$, grow right=18mm] {
        a/$\inc(3)$;
        b/$\inc(4)$ -!- c/$\val():0$ -> d/$\val():3$ -> e/$\val():6$;
        f/$\inc(2)$;
        a -> d;
        b ->[bend right] e;
        c ->[bend right] e;
        f ->[bend right] e;
};
      \subcaption{not local nor monotonic visibility}
    \end{tcolorbox}
  \end{boxes}
  \caption{Examples about basic visibility axioms, considering a single
  counter object.  Operations from each process from left to right, program
  order left out.}
  \label{fig:visibility-examples}
\end{figure*}

\subsection{Closed past}

%

The well-formedness axiom for serializations, $\axiom{ser}{vis}$, states
that all events visible to some event are serialized before that event, but
allows non visible events to also be serialized before. The
\emph{closed past} axiom forbids non visible events to be serialized
before the visible ones.

\begin{definition}[Closed past]
An execution $(H, \po, \vis, \{\ser_i\}_{i \in \ids})$, satisfies
closed past if and only if, for any events $a,c$ in $H$, and
any event $b_i$ from any process $i$:
\[
   a \vis b_i \land c \cnot\vis b_i \implies a \ser_i c
  \tag{$\axiom{ser}{clo}$}
\]
\end{definition}


This means that the past of $b_i$ is closed, when $b_i$ executes, and will not
be changed subsequently, only appended to. Any operation that subsequently
becomes visible will be serialized after all currently visible operations.
Technically, the set of events $\vis b_i$ visible to an event $b_i$, of
process $i$, is a downward closet subset of $H$ totally ordered by $\ser_i$.

This implies that under a sequential specification, and assuming monotonic
visibility, the sequence of operations $S$ visible to a query $q_i$, and
used to answer it, can be reused and extended. When computing the outcome of a
subsequent query $q_i'$, the new sequence of operations $S'$, that became
visible to $q_i'$ but was not visible to $q_i$, can be simply applied after
$S$, conceptually starting from the state resulting from applying $S$. This
means that the execution of any of these operation is final (not tentative),
with no need for them to be (conceptually) re-executed, when computing the
outcome of subsequent operations, as no new operation will be serialized
before any of them.

This applies to most abstractions, where an operation is either a query or an
update, i.e., not only memory (with read/write) but also most data types,
including CRDTs, which classify an operation as either a query or update. But
even for abstractions having operations that are both a query and update (e.g.,
stack, queue), closed past ensures that a single operation execution is
enough, by splitting an operation in two parts: the query part, to calculate
the result, when the operation is issued, and the update part, when it becomes
visible. So, closed past allows a single final operation execution for all
abstractions.

This also means that serialization coincides with the order of
manifestation of visibility in the process, being incrementally built as
execution proceeds, only by appending new events after the already visible
ones, not ``adding to the past''.
Almost all consistency models satisfy this axiom.
Only systems which involve conceptual re-execution do not satisfy it. This is the
case for some eventually consistent systems, motivated by aiming to achieve
convergence by building an \emph{arbitration}:
a global serialization, common to all processes. For these systems to be
\emph{available}, this arbitration needs to be built a posteriori, after some
operations have already executed, by adding to their past. We return to this
issue below.

Figure~\ref{fig:closed-past} presents an example, for a single queue object,
with traditional sequential specification, where closed past is not satisfied.
This execution is valid, assuming the presented serialization for process $j$.
The first $\deq()$ returns $4$, because $\enq(1)$ is not yet visible; 
it only becomes visible when the subsequent $\val()$ is executed. Because
$\enq(1)$ is ordered in the past of the other, already visible enqueues,
$\val()$ returns the surprising value of $[4, 2]$, with $4$ still in the
queue, even though dequeue returned it. This $4$ is subsequently enqueued,
making the final value be $[4, 2, 4]$.

\begin{figure*}
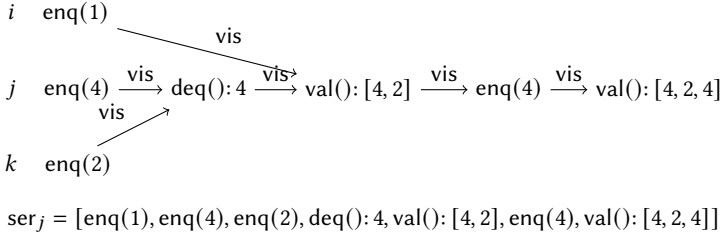

  \small
  \begin{boxes}{1}
    \begin{tcolorbox}
  \tikz \graph [edge label=${\vis}$, grow right=20mm] {
    a/$i\quad\enq(1)$;
    b/$j\quad\enq(4)$ -> c/$\deq():4$ -> d/$\val():{[4, 2]}$ -> e/$\enq(4)$ -> f/$\val():{[4, 2, 4]}$;
    g/$k\quad\enq(2)$;
    a -> d;
    g -> c;
  };
  \[
    {\ser_j} = [\enq(1), \enq(4), \enq(2), \deq():4, \val():[4,2], \enq(4), \val():[4, 2, 4]]
  \]
    \end{tcolorbox}
  \end{boxes}
  \caption{Example without closed past considering a queue object.
  Operations from each process placed in horizontal lines, program order from left to right, left out.
  Transitive visibility, only direct edges drawn. Serialization for process $j$ shown.
  }
  \label{fig:closed-past}
\end{figure*}

It should be noticed that this outcome results from the serial application of
the \emph{effects} of visible operations at each point (the sequential semantics
uses $\op(o)$ for an event $o$, ignoring the result). But the execution, at
each process, namely $j$, is not a serial execution from the point of view of
the process, as the result returned from each operation does not match the
outcome from the corresponding serial execution.

\begin{remark}
A serial execution must have effects and return values consistent with the
sequential specification. This is the same concept that can be used for
transactions, but also for individual
operations~\cite[p.422]{DBLP:conf/ac/Gray78}.
An execution is serial from the point of view of a process if all results of
the queries by the process can be explained by some serial execution. When a
serial execution at one process is built, combining operations from several
processes, the results at other processes are ignored but the results from
queries at the process itself, along the execution, must match the specification.
Reads at other processes are ignored in traditional models with read/write
operations; in our framework, the semantic function for sequential
specifications already ignores the results of other operations.
\end{remark}

In this example, the serial execution would make $\deq()$ return $1$.
Moreover, if process $j$ was running the program:
\[
  \enq(4) ; x = \deq() ; \val() ; \enq(x) ; \val()
\]
where the subsequent enqueue used the result of the previous dequeue, assigned
to $x$, as argument, it would perform $\enq(1)$, making the history itself different.

Rewriting the past, but only for effects, leads to such strange outcomes that
do not occur in classic consistency models where any outcome is explained by a
serial execution.

\subsection{Serial Consistency}

We now define a consistency model slightly stronger than the
combination of the three basic axioms (monotonic visibility, local visibility,
and closed past), that we call \emph{Serial Consistency}. This corresponds to
a generalization for this framework of the weakest model in
\textcite{DBLP:conf/europar/BatallerB97, DBLP:journals/jacm/SteinkeN04}, there
called Local Consistency, where monotonic visibility and closed past are implicitly
assumed. Those frameworks for shared memory, restricted to read/write
operations, are based on \emph{serial views}, without a visibility relation,
in a context where not having monotonic visibility or closed past would be
inconceivable.

Serial consistency means, for sequential specifications, that each process can
explain the local results from a sequence of operations, respecting its own
local program order, obtained by interleaving the effects of remote
operations, in some order.

\begin{definition}[Serial Consistency]
An execution $(H, \po, \vis, (\ser_i)_{i \in \ids})$ satisfies serial
consistency if and only if,
 for any event $o_i$ of any process $i$:
\[
  {\ser_i o_i} = {\vis o_i}
  \tag{$\cons{serial}$}
\]
\end{definition}

This can also be written as $a \ser_i b_i \iff a \vis b_i$, and expresses
that the set of events serialized before some event coincides with those
visible to it. For sequential specifications it means that the sequence of
operations considered to evaluate the result of an operation are exactly those
serialized before it. Because this is true for all operations of a process,
each process can explain the local history by a serial execution of its operations
interleaved with the effects of operations from other processes.

\begin{proposition}
A valid execution $(H, \po, \vis, \{\ser_i\}_{i \in \ids})$ that satisfies
serial consistency also satisfies monotonic visibility, local visibility and
closed past:
\[
  \cons{serial} \implies
  \axiom{vis}{mon} \land
  \axiom{vis}{loc} \land
  \axiom{ser}{clo}
\]
\end{proposition}
\begin{proof}
  1) $\cons{serial} \implies \axiom{vis}{loc}$:
assume $a_i \po b_i$; then, $a_i \ser_i b_i$ (it cannot be the case
that $b_i \ser_i a_i$, as it would imply $b_i \vis a_i$, which would
  contradict $\axiom{vis}{pr}$); therefore, $a_i \vis b_i$.
  2) $\cons{serial} \implies \axiom{vis}{mon}$:
assume $a \vis b_i \po c_i$; then, $a \ser_i b_i$ and
also $b_i \ser_i c_i$ (as above); therefore, $a \ser_i c_i$, which implies
$a \vis c_i$.
  3) $\cons{serial} \implies \axiom{ser}{clo}$:
assume $a \vis b_i \land c \cnot\vis b_i$;
then, $a \ser_i b_i$ and also $c \cnot{\ser_i} b_i$; therefore, $b_i \ser_i c$,
which implies $a \ser_i c$.
\end{proof}

The converse is not true; serial consistency is not equivalent to
the combination of the three basic axioms.
If an operation $a$ is not visible to a local operation $b$,
they forbid serializing $a$ before some already visible operation, but not
before $b$ itself.
This is easily seen by an example, shown in Figure~\ref{fig:non-serial}, for a
single stack object.  This execution satisfies the three basic axioms,
including closed past, which allows $\push(3)$ to be serialized before
$\pop()$, even though it is not visible to it. Closed past would only forbid
$\push(3)$ to be serialized before any other push operation, from $i$, visible
to $\pop()$.  This example shows that even the combination of the three basic
axioms still allows ``strange outcomes'', only prevented by serial
consistency.

Serial consistency is desirable not only for sequential specifications.
For concurrent specifications, it means that there is a serialization that can
explain the order of manifestation of the monotonically increasing visibility
along the process execution, with local events becoming immediately visible,
and no remote events that become visible being serialized before the already
visible ones. It implies the three basic axioms and prevents undesirable
outcomes, that could occur otherwise, making it also relevant for models with
abstractions defined using concurrent specifications.

\begin{figure*}
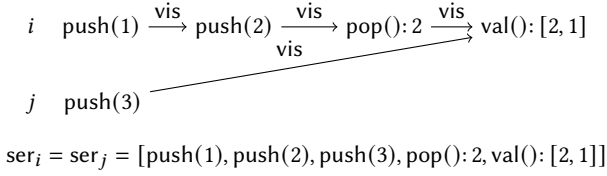

  \small
  \begin{boxes}{1}
    \begin{tcolorbox}
  \tikz \graph [edge label=${\vis}$, grow right=20mm] {
    a/$i\quad\push(1)$ -> b/$\push(2)$ -> c/$\pop():2$ -> d/$\val():{[2,1]}$;
    e/$j\quad\push(3)$;
    e -> d;
  };
  \[
    {\ser_i} = {\ser_j} = [\push(1), \push(2), \push(3), \pop():2, \val():[2,1]]
  \]
    \end{tcolorbox}
  \end{boxes}
  \caption{Example satisfying basic axioms but not serial consistency,
  considering a stack object.  Operations from each process placed in
  horizontal lines, program order from left to right, left out. Transitive
  visibility, only direct edges drawn.}
  \label{fig:non-serial}
\end{figure*}

\section{Classic consistency models}

We now present generalizations, for our framework, of the three more important
classic timeless consistency models. These models were originally presented in
the context of memory (read/write registers). The generalizations can be used
for arbitrary abstractions, not only memory, respect the original
definitions when instantiated for the special case of memory, and respect
current taxonomies (generalizations should generalize, not contradict). As
the classic specifications, each model satisfies the basic serial consistency.

\subsection{Pipelined consistency}

Pipelined consistency is the generalization of PRAM~\cite{lipton1988pram}
(Pipelined RAM) consistency, defined for memory, to general data
abstractions. An execution is PRAM
consistent~\cite{DBLP:conf/spaa/AhamadBJKN93} if the reads, at each process,
are compatible with a process-specific serialization, of all writes and its own
reads, that satisfies program order. I.e., different processes may see
different orders, but all compatible with the program order.

To generalize PRAM to pipelined consistency in our framework, we can start
from two independent guarantees, which can be useful to obtain other models.
The first is having the manifestation of visibility respect program order.

\begin{definition}[Pipelined visibility]
An execution $(H, \po, \vis, \{\ser_i\}_{i \in \ids})$, satisfies
pipelined visibility if and only if, for any events $a$, $b$, $c$ in $H$:
\[
  a \po b \vis c \implies a \vis c \\
  \tag{$\axiom{vis}{pipe}$}
\]
\end{definition}

This guarantee means that if some operation is visible, all previous operations
from the same process are also visible.
Pipelined visibility can be considered dual to monotonic visibility, as
illustrated in Figure~\ref{fig:pipelined-visibility}; one is about future
operations from the same process, the other about past operations.

\begin{figure}
\begin{center}
  \includegraphics[scale=0.8]{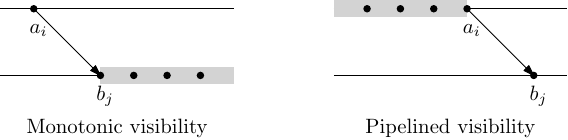}
\end{center}
  \caption{Monotonic visibility: subsequent operations to $b_j$, at process
  $j$ see $a_i$; pipelined visibility: previous operations to $a_i$, from
  process $i$ are visible to $b_j$.}
  \label{fig:pipelined-visibility}
\end{figure}

This is not enough for PRAM, as two
operations from the same process may become visible simultaneously, but be
serialized in the ``wrong'' order. The second guarantee is enforcing the
program order in each serialization.

\begin{definition}[Pipelined serializations]
An execution $(H, \po, \vis, \{\ser_i\}_{i \in \ids})$, satisfies
pipelined serializations if and only if, for any process $i$,
and any events $a$, $b$ in $H$:
\[
  a \po b \implies a \ser_i b
  \tag{$\axiom{ser}{pipe}$}
\]
\end{definition}

These two guarantees together can be called \emph{Pipelining}; they are
the basic ingredients for pipelined consistency, and can be considered a weak
form of pipelined consistency.

\begin{definition}[Pipelining]
An execution $(H, \po, \vis, \{\ser_i\}_{i \in \ids})$, satisfies
pipelining if and only if it satisfies:
\[
  \axiom{vis}{pipe} \land
  \axiom{ser}{pipe}
  \tag{$\axiom{prop}{pipe}$}
\]
\end{definition}

Pipelined consistency can be obtained as the combination of pipelining
together with serial consistency, required for it to be a generalization of
PRAM, satisfying PRAM when instantiated in the context of memory.
The standard PRAM formalization, based exclusively on serializations,
implicitly implies serial consistency in our more general framework.
Pipelining by itself does not imply any of the basic axioms: monotonic
visibility, local visibility or closed past.

\begin{definition}[Pipelined Consistency]
An execution $(H, \po, \vis, \{\ser_i\}_{i \in \ids})$, satisfies
pipelined consistency if and only if it satisfies:
\[
  \axiom{prop}{pipe} \land
  \cons{serial}
  \tag{$\cons{pipe}$}
\]
\end{definition}

\begin{remark}
\label{rem:pram}
Although not commonly realized (if at all), the first axiomatic specification
of PRAM~\cite{DBLP:conf/spaa/AhamadBJKN93}, later adopted by subsequent works,
is more relaxed than the original operational definition of
PRAM~\cite{lipton1988pram}. In the operational definition, reads/writes are to local
memory, with immediate local effect, with writes subsequently propagated to
other processes in FIFO order (considering the name, although not explicitly stated). 
This means that given two writes, $a_i$ at process $i$ and $b_j$ at
process $j$, if $a_i$ is serialized before $b_j$ at $j$, being visible to
$b_j$, then $a_i$ preceded $b_j$ in physical time and the effect of $b_j$
cannot be visible to $a_i$. Therefore, $a_i$ must also be serialized before
$b_j$ at $i$:
\[
  a_i \ser_j b_j \implies a_i \ser_i b_j
  \tag{\axiom{ser}{compatible}}
\]
This axiom, missing from the PRAM axiomatic specification, forbids mutually
incompatible serializations, in which each process serializes the remote
operation before the local one.  The axiomatic specification of PRAM only
constrains individual serializations, but does not have any constraint about
whether all serializations together are physically realizable.  Without this
axiom, the specification fails to prevent physically impossible causality loops, allowing
out-of-thin-air reads.

Consider Figure~\ref{fig:pipe-oota}, representing a history involving two
processes and two memory locations, in which each process does a read of a
location and writes the other location with the value read. In this history
42 comes out-of-thin-air, only justified by a physically impossible causality
loop: the first read of $x$ is explained by a write subsequent to a read which
is explained by a write which is subsequent to that first read of $x$.
It is allowed by the PRAM axiomatic specification, but not the operational
definition. Our framework solves the problem without the need for
this extra axiom: it rejects the execution as not well-formed, as it violates
the physical realizability axiom \axiom{vis}{pr}.

It should be noted that the $\axiom{ser}{compatible}$ axiom about compatible
serializations would apply to most abstractions, like PRAM, with non-blocking
behavior and without cyclic visibility, where operations do not synchronize.
Barrier-like abstractions will not obey it, but they will typically have
concurrent semantics, ignoring serializations.
\end{remark}

\begin{figure*}
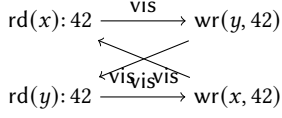

  \small
  \begin{boxes}{1}
  \let\wr=\relax
  \let\rd=\relax
  \auxfun{wr}
  \auxfun{rd}
  \begin{tcolorbox}
      \tikz \graph [edge label=${\vis}$, grow right=25mm] {
        a/{$\rd(x):42$} -> b/{$\wr(y,42)$} ;
        c/{$\rd(y):42$} -> d/{$\wr(x,42)$} ;
        b -> c;
        d -> a;
      };
  \end{tcolorbox}
  \end{boxes}
  \caption{History acceptable by the PRAM axiomatic specification but not the original
  operational definition, only explainable by a causality loop. The execution,
  in our framework, is not well-formed, as it violates \axiom{vis}{pr}.}
  \label{fig:pipe-oota}
\end{figure*}

\subsection{Causal consistency}

Causal Consistency was introduced in the context of memory, i.e., read/write
registers, as usual for most consistency models, in what was called
\emph{Causal Memory}, briefly introduced
by~\textcite{DBLP:conf/icdcs/HuttoA90} and specified axiomatically
by~\textcite{DBLP:journals/dc/AhamadNBKH95}.

It was conceived to be something weaker that sequential consistency (as each
process can see a different serialization), but stronger than PRAM (pipelined
consistency). It guarantees that each serialization provides visibility, not
only of events in the past for each process, but all that are potential causal
dependencies, transitively.

As for pipelined consistency, we can start by two guarantees. The first is
that visibility includes all events in the semantic causal past, given by the
causality relation (semantic happens-before), which we recall was as defined
as:
$
{\hb} \defeq (\po \union \vis)^+
$.

\begin{definition}[Causal visibility]
An execution $(H, \po, \vis, \{\ser_i\}_{i \in \ids})$, satisfies
causal visibility if and only if, for any events $a$, $b$ in $H$:
\[
  a \hb b \implies a \vis b \\
  \tag{$\axiom{vis}{causal}$}
\]
\end{definition}

Visibility includes and is included in happens-before, from its
definition; it follows trivially that:

\begin{proposition}
An execution $(H, \po, \vis, \{\ser_i\}_{i \in \ids})$ satisfies causal
visibility, if and only if
  \[{\vis} = {\hb} \]
\end{proposition}

The second guarantee is that serializations also respect causality, relevant
for classic sequential specifications. (Concurrent specifications ignore
serializations.)

\begin{definition}[Causal serializations]
An execution $(H, \po, \vis, \{\ser_i\}_{i \in \ids})$, satisfies
causal serializations if and only if, for any process $i$,
and any events $a$, $b$ in $H$:
\[
  a \hb b \land b \not\hb a \implies a \ser_i b
  \tag{$\axiom{ser}{causal}$}
\]
\end{definition}

These two guarantees together can be called \emph{Causality}; they
are the basic ingredients for causal consistency, or a weak form of
causal consistency.

\begin{definition}[Causality]
An execution $(H, \po, \vis, \{\ser_i\}_{i \in \ids})$, satisfies
causality if and only if it satisfies:
\[
  \axiom{vis}{causal} \land
  \axiom{ser}{causal}
  \tag{$\axiom{prop}{causal}$}
\]
\end{definition}

\begin{remark}
A question arises: why not define the causality relation simply as the transitive
closure of visibility, to have causal visibility simply ensure transitive
visibility? I.e., $a \vis^+ b \implies a \vis b$.
For classic models, given that local visibility
holds, the two definitions coincide, as $a \po b \implies a \vis b$. 
But in models such as Prefix Consistency, discussed in
Section~\ref{sec:prefix-consistency}, that is not the case. Using
happens-before which mimics the classic one by Lamport, ensures that
potential causality is respected, even if operations do not become immediately
visible to the process issuing them. It also ensures that causality is
stronger than pipelining.
\end{remark}

\begin{proposition}
If some execution $(H, \po, \vis, \{\ser_i\}_{i \in \ids})$ satisfies
causality, then it also satisfies pipelining.
\end{proposition}
\begin{proof}
  If $a \po b \vis c$ then $a \hb c$; by $\axiom{vis}{causal}$, $a \vis c$,
  satisfying $\axiom{vis}{pipe}$.
  If $a \po b$ then $a \hb b$ and $b \not\hb a$; by $\axiom{ser}{causal}$, $a
  \ser_i b$ at any process $i$, satisfying $\axiom{ser}{pipe}$ .
\end{proof}

Causal consistency can be obtained as the combination of causality together
with serial consistency.

\begin{definition}[Causal Consistency]
An execution $(H, \po, \vis, \{\ser_i\}_{i \in \ids})$, satisfies
causal consistency if and only if it satisfies:
\[
  \axiom{prop}{causal} \land
  \cons{serial}
  \tag{$\cons{causal}$}
\]
\end{definition}

Requiring serial consistency is needed for the same reason as for pipelined
consistency: to have causal consistency as a generalization of causal memory
for arbitrary data types, satisfying it when instantiated in the context of
memory. The causal memory formalization, based only on serializations, without
a visibility relation, implicitly implies serial consistency in our framework.
Causality by itself implies monotonic and local visibility but not closed past.

\begin{remark}
\label{rem:cm-cl}
As for PRAM, the axiomatic specification of causal memory is lacking an axiom,
and allows histories only explainable by physically impossible causality loops
when there are several writes of the same value to the same location.
This is a technical issue having to do with 1) use of a writes-into relation,
which only relates a single write as the cause of a read, to define a
causality order, together with 2) serializations are allowed in which
another write of the same value to the same location, serialized after,
explains the read. This allows the real cause of a read not to be accounted
for in the causality order, allowing histories that can only be explained by
physically impossible causality loops of the real information propagation.

An example is shown in Figure~\ref{fig:causal-oota}, where the last read of
each process, by not returning $2$ (which had already overwritten the $1$
previously written), is explained by the last write of the other process.
In the causal memory specification, these last writes are used to explain the
reads in the serializations but not accounted for in the causality order.
In our framework, those writes must be visible to the reads to explain
them (otherwise they would return 2). But doing so makes them accounted for in the
happens-before relation, which renders the execution not admissible as not
being well-formed, as it violates the physical realizability axiom \axiom{vis}{pr}.
Our framework not only is more general, but solves this problem when instantiated
for memory.
\end{remark}

\begin{figure*}
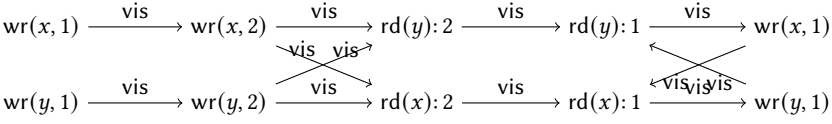

  \small
  \begin{boxes}{1}
  \let\wr=\relax
  \let\rd=\relax
  \auxfun{wr}
  \auxfun{rd}
  \begin{tcolorbox}
      \tikz \graph [edge label=${\vis}$, grow right=25mm] {
  a/{$\wr(x,1)$} -> b/{$\wr(x,2)$} -> c/{$\rd(y):2$} -> d/{$\rd(y):1$} -> e/{$\wr(x,1)$} ;
  f/{$\wr(y,1)$} -> g/{$\wr(y,2)$} -> h/{$\rd(x):2$} -> i/{$\rd(x):1$} -> j/{$\wr(y,1)$} ;
  b -> h;
  g -> c;
  e -> i;
  j -> d;
      };
  \end{tcolorbox}
  \end{boxes}
  \caption{History acceptable by the causal memory specification, where the
  final read of each process can only be explained by a causality loop not
  captured in the causality order.
  The execution, in our framework, with the visibility edges from the final
  writes to reads, to justify the results, is not well-formed, as it violates
  \axiom{vis}{pr}.}
  \label{fig:causal-oota}
\end{figure*}

It is easy to see that causal consistency is stronger than (implies)
pipelined consistency, as causality is stronger than pipelining.

\subsection{Sequential consistency}

Sequential consistency is when there is some global total order
respecting program order, which is seen by all processes, defining both
visibility and each serialization. This makes the outcome equivalent to a
sequential execution at a single process.

\begin{definition}[Sequential Consistency]
An execution $(H, \po, \vis, \{\ser_i\}_{i \in \ids})$, satisfies
sequential consistency if and only if $\vis$ is a strict total order on $H$
such that, for any process $i$:
\[
  {\vis} \supseteq {\po} \land {\ser_i} = {\vis}
  \tag{$\cons{seq}$}
\]
\end{definition}

It is easy to see that sequential consistency is stronger than (implies)
causal consistency.

\begin{proposition}
If some execution $(H, \po, \vis, \{\ser_i\}_{i \in \ids})$ satisfies sequential
consistency, then it also satisfies causal consistency.
\end{proposition}
\begin{proof}
  As ${\vis} \supseteq {\po}$ and $\vis$ is a strict total order, then ${\hb} =
  {\vis}$
  and so $\axiom{vis}{causal}$ holds.
  As ${\ser_i} = {\vis} = {\hb}$, for each process $i$, both $\axiom{ser}{causal}$
  and $\cons{serial}$ hold.
\end{proof}

As an example of what cyclic visibility allows, a relaxed variant of
sequential consistency can be defined: \emph{Set-Sequential Consistency}. This
is the analogue for timeless consistency models of
set-linearizability~\cite{DBLP:conf/podc/Neiger94}, and expresses that there
is a total order of equivalence classes of concurrent operations.

\begin{definition}[Set-Sequential Consistency]
An execution $(H, \po, \vis, \{\ser_i\}_{i \in \ids})$, satisfies
set-sequential consistency if and only if $\vis^=$ is a total preorder on $H$
such that, for any process $i$, and any events $a$, $b$ in $H$:
\[
  {\vis} \supseteq {\po} \land (a \vis b \land b \not\vis a) \implies a \ser_i b
  \tag{\cons{set-seq}}
\]
\end{definition}

\section{Convergence}
\label{sec:convergence}

Strong consistency models, such as sequential consistency, give us
convergence. Even if the implementation involves several replicas, the outcome
is equivalent to as if there is a single replica. I.e., all replicas are
always converged, and convergence is not even explicitly mentioned.

For weaker models convergence becomes an issue, but as achieving convergence
is a liveness property, convergence by itself has not been incorporated in
consistency models, let alone as a safety property. Some
taxonomies~\cite{DBLP:journals/csur/ViottiV16} use the notion of
\emph{Eventual Consistency} or \emph{Strong Eventual Consistency}, but
these involve liveness, unlike criteria that are purely about safety, as
desirable for consistency models. Moreover, eventual consistency only allows
reasoning about infinite histories.
Even~\textcite{MADahlin2011}, devoted to convergence, mentions it a as liveness
property. Specific models have been presented, such as \emph{Causal
Convergence}~\cite{DBLP:conf/ppopp/PerrinMJ16}, but based on eventual
consistency (therefore involving liveness) and overly specific (assuming a
global arbitration). Convergence as a standalone safety criterion, to be
combined with others, is missing from taxonomies of consistency models.

\subsection{Convergence}

We present here convergence as a safety property, defining \emph{Convergent
Consistency}, or simply \emph{Convergence} for short.
Eventual consistency involves liveness, being applicable only to infinite
executions, but not to finite executions corresponding to actual runs of some
distributed system. The key insight is not to use eventual consistency as
criterion in models, but split it into the two properties that eventually
consistent systems ensure, according to the original detailed introduction of
the term eventual
consistency~\cite{DBLP:conf/pdis/TerryDPSTW94,DBLP:conf/sosp/TerryTPDSH95}.
One is about liveness and the other exclusively about safety. Rephrasing those
properties in the current framework terminology:
\begin{itemize}
  \item eventual visibility: each operation becomes visible eventually; only a
    finite number of operations can have it not visible (liveness).
  \item convergence: processes having the same set of visible operations
    have equivalent states, returning the same result from the same operation
    (safety).
\end{itemize}

This second property is exactly the perfect criterion for convergence
for consistency models:
exclusively about safety and quite general (not requiring an arbitration).

\begin{definition}[Convergence]
An execution $(H, \po, \vis, \{\ser_i\}_{i \in \ids})$, satisfies
convergence if and only if it satisfies:
\[
  \op(a) = \op(b) \land {\vis a} =  {\vis b} \implies \re(a) = \re(b)
  \tag{$\axiom{res}{conv}$}
\]
\end{definition}

This means that different processes that have the same set of visible
operations will return the same result, even if the operations became visible
in different orders and the processes have different serializations.

The first two executions in Figure~\ref{fig:convergence-examples} are with
and without convergence. A counter (\ref{fig:convergence-examples-counter}),
with increment being commutative satisfies convergence, returning the same
values for the same set of increments. This holds both for a concurrent and a
sequential counter specification.
In the queue execution (\ref{fig:convergence-examples-queue}), assuming closed
past, the remote enqueue must be serialized after the local enqueue, leading
to different outcomes of the final queries.  So, the same visible operations
lead to different outcomes, not satisfying convergence.

\begin{figure*}
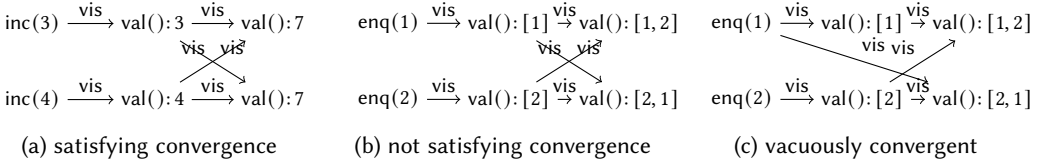

  \footnotesize
  \begin{boxes}{3}
    \begin{tcolorbox}
      \tikz \graph [edge label=${\vis}$, grow right=16mm] {
        a/$\inc(3)$ -> b/$\val():3$ -> c/$\val():7$;
        d/$\inc(4)$ -> e/$\val():4$ -> f/$\val():7$;
        b -> f;
        e -> c;
};
      \subcaption{satisfying convergence}
      \label{fig:convergence-examples-counter}
    \end{tcolorbox}
    \begin{tcolorbox}
      \tikz \graph [edge label=${\vis}$, grow right=16mm] {
        a/$\enq(1)$ -> b/$\val():{[1]}$ -> c/$\val():{[1,2]}$;
        d/$\enq(2)$ -> e/$\val():{[2]}$ -> f/$\val():{[2,1]}$;
        b -> f;
        e -> c;
};
      \subcaption{not satisfying convergence}
      \label{fig:convergence-examples-queue}
    \end{tcolorbox}
    \begin{tcolorbox}
      \tikz \graph [edge label=${\vis}$, grow right=16mm] {
        a/$\enq(1)$ -> b/$\val():{[1]}$ -> c/$\val():{[1,2]}$;
        d/$\enq(2)$ -> e/$\val():{[2]}$ -> f/$\val():{[2,1]}$;
        a -> f;
        e -> c;
};
      \subcaption{vacuously convergent}
      \label{fig:defeating-convergence-queue}
    \end{tcolorbox}
  \end{boxes}
  \caption{Executions a) with a counter, convergent; b) a queue, with closed
  past, not convergent); c) alternative execution for history in b), vacuously
  convergent.
  Operations from each process from left to right, program order left out.
  Transitive visibility, only direct edges drawn.}
  \label{fig:convergence-examples}
\end{figure*}

Deterministic concurrent specifications satisfy convergence trivially, given that
the semantic function ignores serialization ($\sfR{o,\vis,\ser_i} = \sfC{o,
\vis}$). Two events with the same operation and the same set of visible
events must return the same result in a valid execution.

\subsection{Arbitration}

For sequential specifications, a way to ensure convergence is to
have the same serialization at all processes. Some frameworks for eventual
consistency~\cite{DBLP:journals/ftpl/Burckhardt14,DBLP:conf/popl/BurckhardtGYZ14},
assume the existence of an \emph{arbitration}, a global total order (or a
union of total orders, but global anyway). Our framework includes per-process
serializations, possibly different in general. We say that an execution where
all serializations are the same satisfies \emph{arbitration}.


\begin{definition}[Arbitration]
An execution $(H, \po, \vis, \{\ser_i\}_{i \in \ids})$, satisfies
arbitration if all serializations are identical,
i.e., for any processes $i$ and $j$:
\[
  {\ser_i} = {\ser_j}
  \tag{$\axiom{ser}{arb}$}
\]
\end{definition}

From its definition, arbitration implies convergence, for any kind
of specification in our framework.

\begin{proposition}
If some execution $(H, \po, \vis, \{\ser_i\}_{i \in \ids})$ satisfies
arbitration, then it also satisfies convergence.
\end{proposition}

\begin{proof}
A valid execution must satisfy result validity: $\re(o) =
\sfR{o,\vis,\ser_i}$, for each event $o$.  If all $\ser_i$ are the same for
each process $i$ (arbitration), the semantic function will give the same
result for two events $a$ and $b$ with the same operation and set of visible
events, making $\re(a) = \re(b)$.
\end{proof}

From the definition it also follows trivially that sequential consistency is
stronger than arbitration.

\begin{proposition}
If some execution $(H, \po, \vis, \{\ser_i\}_{i \in \ids})$ satisfies sequential
consistency, then it also satisfies arbitration.
\end{proposition}

So, arbitration can be used to achieve convergence, for sequential
specifications. Sequential consistency is a special case, but weaker
models satisfying arbitration will also allow convergence.
Deterministic concurrent specifications do not need arbitration as they will
always satisfy convergence.

\begin{proposition}
Some execution $(H, \po, \vis, \{\ser_i\}_{i \in \ids})$ satisfies sequential
consistency if and only if it satisfies serial consistency and arbitration:
\[
  \cons{seq} \iff \cons{serial} \land \axiom{ser}{arb}.
\]
\end{proposition}
\begin{proof}
As sequential consistency implies causal consistency, therefore serial
consistency, and arbitration, it remains to show the reverse
implication, i.e., that that these imply: 
  1) ${\vis} \supseteq {\po}$; this is local visibility, implied by serial
  consistency.
  2) ${\ser_i} = {\vis}$, for each process $i$;
there are two possible cases for pairs of events, an event $a$ and either:
an event $b_i$ from process $i$, in which case,
from serial consistency, $a \ser_i b_i \iff a \vis b_i$; or an event
$b_j$ from some other process $j$; then, $a \ser_j b_j \iff a \vis b_j$; but
as all serializations are identical, it also implies that
$a \ser_i b_j \iff a \vis b_j$.
\end{proof}

\section{Existential versus universal quantification of executions}

The previous sections have defined consistency criteria for valid executions.
But the ultimate aim is to express whether a history satisfies a consistency
model, and to be able to validate implementations. An implementation of an
abstraction aiming to satisfy a consistency model is correct if any possible
history allowed by the implementation satisfies it.

For traditional consistency models, a history is defined as satisfying the
model is there is some valid execution justifying the history that satisfies
the model. This amounts to existentially assuming such an execution.

\begin{definition}[Existential satisfaction for histories]
A history $H$, with program order $\po$, satisfies existentially a consistency
model $C$ if there exists some valid execution $(H, \po, \vis, \{\ser_i\}_{i
\in \ids})$, that satisfies $C$.
\end{definition}

This definition is appropriate to define whether a history satisfies any of the
consistency criterion discussed above, with one exception: convergence.
The way this criterion is defined, while suitable considering specific
executions, renders the existential quantification definition useless, as it
would cause most histories to be vacuously considered as satisfying it.

The reason is that normally there are many valid executions that justify a
given history, differing in visible edges or serializations. For most
histories, it would be possible to pass the criterion by ``feeding'' different
sets of visible operations to $a$ and $b$, whenever $\op(a) = \op(b)$. This
can be done, e.g., making redundant operations, such as pure queries, visible
to $a$ but not $b$ or vice versa.
This frees $a$ and $b$ from the requirement to return the same result.

An alternative to the queue execution in
Figure~\ref{fig:convergence-examples-queue}, justifying the same history,
could be existentially assumed, in which different queries see different sets
of visible operations. This is shown in
Figure~\ref{fig:defeating-convergence-queue}.  As no two $\val$ from different
processes have the same set of visible operations, the constraint is vacuously
met. This would make the history be considered as satisfying convergent
consistency, under the existential definition, something clearly not
desirable.

It could be thought that the difficulty results from allowing pure queries to be
considered visible. But not allowing them to be visible not only would ruin
the framework elegance (by not allowing simple rules, such as ${\vis} = {\hb}$ for
causal visibility) but would not solve the problem. The reason is that
different sets of visible operations can be formed through redundant
updates, updates that cancel each other, or equivalent updates from different
processes (that are different events).

Two examples are shown in Figure~\ref{fig:defeating-convergence-examples}. On
the left, an execution involving a queue, with two extra processes issuing an
enqueue each, identical to the other two enqueues. In this execution no two
$\val$ queries have the same set of visible events, therefore vacuously satisfying
convergence. On the right, an execution for a stack, resembling
Figure~\ref{fig:convergence-examples-queue} for the two
topmost processes, with a third process issuing a push and pop canceling
each other, while contributing to the visible events of the second process,
which results in the identical outcome of vacuous satisfaction of convergence.

\begin{figure*}
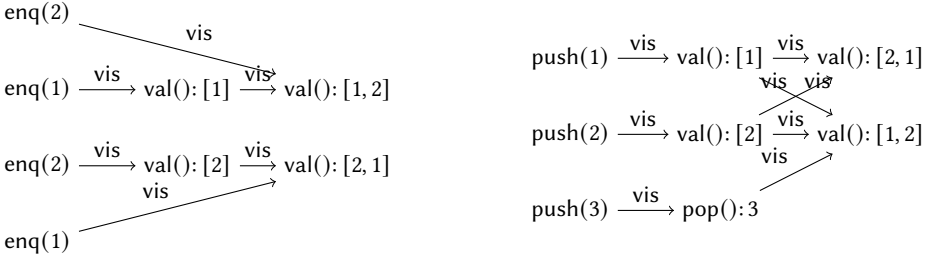

  \small
  \begin{boxes}{2}
    \begin{tcolorbox}
      \tikz \graph [edge label=${\vis}$, grow right=20mm] {
        g/$\enq(2)$;
        a/$\enq(1)$ -> b/$\val():{[1]}$ -> c/$\val():{[1,2]}$;
        d/$\enq(2)$ -> e/$\val():{[2]}$ -> f/$\val():{[2,1]}$;
        h/$\enq(1)$;
        g -> c;
        h -> f;
};
    \end{tcolorbox}
    \begin{tcolorbox}
      \tikz \graph [edge label=${\vis}$, grow right=20mm] {
        a/$\push(1)$ -> b/$\val():{[1]}$ -> c/$\val():{[2,1]}$;
        d/$\push(2)$ -> e/$\val():{[2]}$ -> f/$\val():{[1,2]}$;
        g/$\push(3)$ -> h/$\pop():{3}$;
        b -> f;
        e -> c;
        h -> f;
};
    \end{tcolorbox}
  \end{boxes}
  \caption{Executions vacuously satisfying convergence, for queue
  and stack, with closed past.
  Operations from each process from left to right, program order left out.
  Transitive visibility, only direct edges drawn.}
  \label{fig:defeating-convergence-examples}
\end{figure*}

For each of these two histories, an alternative execution exists
not satisfying convergence, similar to the one in
Figure~\ref{fig:convergence-examples-queue} in what concerns the
processes issuing $\val$, with the operations from the extra processes not
visible. Such executions could correspond to a run with the extra process(es)
partitioned, with the two final $\val$ having identical sets of visible
operations while having different results.
It is desirable that this history is considered not convergent, as
some executions that explain it violate convergence. The way to achieve it is to
consider a history as satisfying convergence only if all valid
executions for the history satisfy it. This amounts to universally quantifying
over valid executions.

\begin{definition}[Universal satisfaction for histories]
A history $H$, with program order $\po$, satisfies universally a consistency
model $C$ if any valid execution $(H, \po, \vis, \{\ser_i\}_{i \in \ids})$,
satisfies $C$.
\end{definition}

Universal satisfaction will thus be the appropriate criterion for classifying a
history as satisfying convergence. It means that
whatever the way the history can be justified (i.e., no matter the way
that operations became visible and were ordered) the convergence requirement
holds.
While the other criteria impose constraints relating program order, visibility
and serializations, for which some solution is sought, convergence has a
different nature, resulting from having ${\vis a} = {\vis b}$ in the
antecedent of an implication. Universal satisfaction prevents undesired vacuous
solutions that contradict the desired classification.

We conjecture that universal satisfaction is not needed for the purpose of
checking implementations: as an implementation must be correct for every
history that can be produced, for each vacuously accepted history, another
history without the extraneous events would also be produced, which would not
be vacuously accepted, failing for a wrong implementation. But even if the
conjecture is true, it would not solve the problem of classifying histories
appropriately regarding convergence, for which universal satisfaction is the
solution.

%

\section{Non-serial consistency models with arbitration}

Replica convergence for sequential specifications can be achieved
by establishing an arbitration, which defines the sequence in which operations
are applied. To ensure high-availability two approaches have been pursued,
which we describe next:
one forgoes closed past while keeping monotonic and local visibility; 
the other forgoes local visibility while keeping monotonic visibility and closed past.
A weaker variant could keep only monotonic visibility and forgo the other two
basic axioms, but there is no good reason to do so, and we are not aware of
existing models that do it.

\subsection{Replay Consistency}

To ensure high-availability, one approach, followed by
Bayou~\cite{DBLP:conf/sosp/TerryTPDSH95}, is to make both local operations and
remote operations that have been received immediately visible (applied), even if
other operations could arrive in the future, and be arbitrated before them.
When an operation $o$ arrives, all operations already applied and arbitrated
after $o$ are undone and then reapplied after applying $o$.
OpSets~\cite{DBLP:journals/corr/abs-1805-04263} and
ECROs~\cite{DBLP:journals/pacmpl/PorreFPB21} follow a similar approach.
In our framework this amounts to forgoing closed past while keeping the other
two basic axioms. We name this model \emph{Replay Consistency}.

\begin{definition}[Replay Consistency]
An execution $(H, \po, \vis, \{\ser_i\}_{i \in \ids})$, satisfies
replay consistency if and only if it satisfies monotonic visibility, local
visibility, and arbitration:
\[
  \axiom{vis}{mon} \land
  \axiom{vis}{loc} \land
  \axiom{ser}{arb}
  \tag{$\cons{replay}$}
\]
\end{definition}

\begin{proposition}
If some execution satisfies replay consistency, then it also satisfies 
pipelined serializations.
\end{proposition}

\begin{proof}
For any events $a$ and $b$ from the same process $i$, if $a \po b$ then $a \vis
b$ ($\text{VIS}_{LOC}$), which implies $a \ser_i b$ ($\text{SER}_{VIS}$);
therefore, $a \ar b$, which means that for any other process $j$, $a \ser_j b$.
\end{proof}

However, replay consistency does not imply pipelined visibility (and
therefore, pipelining), as the lack of closed past allows non-visible events to
be serialized before visible ones.
A stronger variant of replay consistency can be defined, satisfying
pipelining and an even stronger one satisfying causality.

\begin{definition}[Pipelined Replay Consistency]
An execution $(H, \po, \vis, \{\ser_i\}_{i \in \ids})$, satisfies
pipelined replay consistency if and only if it satisfies replay consistency and
pipelining:
\[
  \cons{replay} \land
  \axiom{prop}{pipe}
  \tag{$\cons{preplay}$}
\]
\end{definition}

\begin{definition}[Causal Replay Consistency]
An execution $(H, \po, \vis, \{\ser_i\}_{i \in \ids})$, satisfies
causal replay consistency if and only if it satisfies replay consistency and
causality:
\[
  \cons{replay} \land
  \axiom{prop}{causal}
  \tag{$\cons{creplay}$}
\]
\end{definition}

We remark that undoing and re-executing operations is conceptual, and may
be optimized in practice, depending on the data abstraction. For memory
(read/write operations) no re-execution is necessary, being enough to simply
ignore writes being received if they are arbitrated before an already visible
write for the same memory location.

\subsection{Prefix Consistency}
\label{sec:prefix-consistency}

An alternative approach, to avoid having to replay operations and have
immediate finality, once an operations is executed, is to keep closed past.
An operation is only made visible once no more operations can
be arbitrated before it, which implies that local operations may need to have
their visibility postponed for some time, i.e., forgoing local visibility.
We name the resulting model \emph{Prefix Consistency}.

\begin{definition}[Prefix Consistency]
An execution $(H, \po, \vis, \{\ser_i\}_{i \in \ids})$, satisfies
prefix consistency if and only if it satisfies monotonic visibility, closed past, and
arbitration:
\[
  \axiom{vis}{mon} \land
  \axiom{ser}{clo} \land
  \axiom{ser}{arb}
  \tag{$\cons{prefix}$}
\]
\end{definition}

By forgoing local visibility, prefix consistency does not necessarily imply
pipelined serializations, and therefore, pipelining.
A stronger variant of prefix consistency can be defined, satisfying
pipelining.

\begin{definition}[Pipelined Prefix Consistency]
An execution $(H, \po, \vis, \{\ser_i\}_{i \in \ids})$, satisfies
pipelined prefix consistency if and only if it satisfies prefix consistency and
pipelining:
\[
  \cons{prefix} \land
  \axiom{prop}{pipe}
  \tag{$\cons{pprefix}$}
\]
\end{definition}

%

An even stronger model can be defined, satisfying causality.

\begin{definition}[Causal Prefix Consistency]
An execution $(H, \po, \vis, \{\ser_i\}_{i \in \ids})$, satisfies
causal prefix consistency if it satisfies prefix consistency and
causality:
\[
  \cons{prefix} \land
  \axiom{prop}{causal}
  \tag{$\cons{cprefix}$}
\]
\end{definition}

We remark that causal prefix consistency satisfies arbitration, monotonic
visibility, local visibility (implied by causality), and closed past.
The three basic axioms together are weaker than serial consistency.
Together with arbitration they remain weaker than sequential
consistency, as exemplified in Figure~\ref{fig:non-serial}, for which arbitration also
holds, but the execution is not sequentially consistent.
Therefore, causal prefix consistency is weaker than sequential consistency.
In practice, it is not an interesting model, as it does not allow wait-free
implementations for general abstractions, as we will show next, while still
allowing ``strange'' outcomes.

\section{The CAL trilemma for wait-free systems}

The CAP theorem~\cite{DBLP:conf/hotos/FoxB99,DBLP:journals/sigact/GilbertL02}
implies that if a system aims to remain always Available (i.e., operations
returning a result) even during network Partitions, strong Consistency (more
exactly, linearizability) cannot be ensured. But it does not help in
understanding possible tradeoffs in the AP region of the design space, i.e.,
highly available systems.

In a highly available system operations must return without waiting
unboundedly for messages (as there could be a partition lasting an arbitrary
amount of time). This means that operations must be able to return after some
finite steps without depending on other processes.  These systems are, therefore,
wait-free~\cite{DBLP:journals/toplas/Herlihy91}.  Finite timeouts can be
used, but to achieve essentially ``zero latency'', a common approach used by
CRDTs is not to wait for any message, but simply use the local state at
invocation time (i.e., the operations already visible) to compute the
result.

We now show that AP systems, satisfying Monotonic visibility, cannot ensure
simultaneously: 1) Closed past, 2) Arbitration and 3) Local visibility.
This is expressed in the CLAM theorem that we present below. Given that
monotonic visibility is fundamental and should not be forgone, the CAL trilemma is
about which of the other three guarantees (C, A, or L) to forgo.

\subsection{Phantom visibility, overriding, and irredundant operations}

There is a technical issue that must be addressed. While the definition of
visibility states ``if the effect is visible'', there is considerable freedom
when existentially assuming executions whether to consider some event as
visible or not, if the effect is irrelevant or redundant, regardless of it
being ``really visible'' due to physical propagation of information.  This
will make it seem possible to overcome the CAL trilemma, implying some care in
formulating the CLAM Theorem below.

Consider the execution involving a register in
Figure~\ref{fig:defeating-cal-phatom-visibility-register}.
Not only it satisfies serial consistency but also arbitration (an arbitration
can be defined combining the fragments of serializations shown), overcoming the trilemma.
The visibility edge from event $d$ to $a$ (in blue) was existentially assumed
so that $d$ can be serialized before $a$ in process $i$, without violating
closed past, allowing identical serializations for both processes, i.e., arbitration.
But this visibility edge carries no information, being redundant; the
operations give the same result with or without it; we refer to it as
\emph{phantom visibility}.

\begin{figure*}
  \small
  \begin{boxes}{1}
  \let\wr=\relax
  \let\rd=\relax
  \auxfun{wr}
  \auxfun{rd}
  \begin{tcolorbox}
      \tikz \graph [edge label=${\vis}$, grow right=20mm] {
        a/$i\quad\wr^a(1)$ -> b/$\rd^b():1$ -> c/$\rd^c():1$;
        d/$j\quad\wr^d(2)$ -> e/$\rd^e():2$ -> f/$\rd^f():1$;
        a -> f;
        d ->[blue] a;
      };
  \[
    \begin{split}
      {\ser_i} & \supset [\wr^d(2), \wr^a(1), \rd^b():1, \rd^c():1 ] \\
      {\ser_j} & \supset [\wr^d(2), \rd^e():2, \wr^a(1), \rd^f():1 ]
    \end{split}
  \]
    \tcblower
    \subcaption{register satisfying serial consistency and arbitration}
    \label{fig:defeating-cal-phatom-visibility-register}
  \end{tcolorbox}
  \begin{tcolorbox}
      \tikz \graph [edge label=${\vis}$, grow right=20mm] {
        a/{$i\quad\wr^a(x,1)$} -> b/{$\rd^b(x):1$} -> c/{$\rd^c(y):0$}
        -> d/{$\rd^d(y):2$};
        e/{$j\quad\wr^e(y,2)$} -> f/{$\rd^f(y):2$} -> g/{$\rd^g(x):0$}
        -> h/{$\rd^h(x):1$};
        a -> h;
        e -> d;
      };
  \[
    \begin{split}
      {\ser_i} & \supset
      [\wr^a(x,1), \rd^b(x):1, \rd^c(y):0, \wr^e(y,2), \rd^d(y):2] \\
      {\ser_j} & \supset
      [\wr^e(y,2), \rd^f(y):2, \rd^g(x):0, \wr^a(x,1), \rd^h(x):1]
    \end{split}
  \]
    \tcblower
    \subcaption{memory satisfying serial consistency but not arbitration}
    \label{fig:defeating-cal-phatom-visibility-memory}
  \end{tcolorbox}
  \end{boxes}
  \caption{Executions satisfying serial consistency and either (a) satisfying
  arbitration or (b) not satisfying arbitration. Transitive visibility, only direct edges
  drawn. Serializations showing own operations and writes.
  Phantom edges added to maintain visibility transitive assumed for (a) and not
  explicitly drawn.}
  \label{fig:defeating-cal-phatom-visibility}
\end{figure*}

A phantom visibility edge provides no actual information and when added to an
execution results in an alternative execution that explains the same history.
They are useful to existentially assume executions aiming to satisfy some desired
constraint, for some consistency model, allowing elegant formulations.

\begin{definition}[Phantom visibility]
Given a valid execution $E = (H, \po, \vis, \{\ser_i\}_{i \in \ids})$,
a visibility edge $e$ not present in $\vis$ is phantom if adding it to $\vis$ results in
  another execution $E' = (H, \po, {\vis} \union \{e\}, \{\ser_i\}_{i \in
  \ids})$, for the same history $(H, \po)$, which is also valid.
\end{definition}

In the register example
(Figure~\ref{fig:defeating-cal-phatom-visibility-register}), suppose that the
system was partitioned until after the first query from each process. Those
reads would return, 1 and 2 respectively, seeing the effect of the local
write. When the write from $j$ arrives at $i$, this process
``pretends'' that it was already visible and serialized before $a$,
effectively ignoring it.

This ``pretending'' is possible for this abstraction (single register),
because every update overrides all previous updates, making them redundant. It
has a sequential semantics in which only the last update serialized before a
read matters. This makes it possible to pretend, when an operation arrives,
that it was already visible in the past, without affecting history and without implying
re-execution. But the single register abstraction is the exception. For
practically all data abstractions, including memory, such pretense will not be
possible.

Consider the history involving memory in
Figure~\ref{fig:defeating-cal-phatom-visibility-memory}.
Unlike for the register history, for this one no execution can be found
that satisfies both serial consistency and arbitration, by adding phantom
edges to the execution shown. (Otherwise, it would be sequentially consistent,
which is clearly not the case. For this example, no possible result value
for the final read of each process could make it sequentially consistent.)
As each write does not override the other, as they are to different memory
locations, there is no way to pretend that one was already visible, and
serialize it before the other, as in the register example, without altering
the outcome and rendering the execution invalid.


The concept of operation overriding is, therefore, essential to the
definition of the CLAM theorem, but it will be more useful to define the
``opposite'', operation irredundance, for operations that do not override each
other.

\newcommand\viewers[1]{\af{lv}(#1)}

\begin{definition}[Local viewers]
The local viewers of an event $a$ in an execution $E = (H, \po, \vis,
  \{\ser_i\}_{i \in \ids})$, is the set of events
  $\viewers a \defeq \{b | a \po b \land a \vis b\}$,
i.e., those from the same process as $a$ for which $a$ is visible.
\end{definition}

\begin{definition}[Irredundant pair of updates]
Two updates $a_i$ and $b_j$ form an irredundant pair of updates if there
exists a history $(H, \po)$ containing them for which, given any valid
execution $E = (H, \po, \vis, \{\ser_i\}_{i \in \ids})$ that explains it,
satisfying $\forall c \in \viewers{b_j} \cdot a_i \not\vis c $,
$\forall c \in \viewers{a_i} \cdot b_j \not\vis c $,
and $a_i \ser_j b_j$ or $b_j \ser_i a_i$,
adding one visibility edge from $a_i$ to each event in $\viewers{b_j}$,
if $a_i \ser_j b_j$,
  or from $b_j$ to each event in $\viewers{a_i}$,
if $b_j \ser_i a_i$,
will produce an execution
contradicting result validity ($\axiom{res}{val}$).
\end{definition}

We call a history that shows a pair of updates to be irredundant an
\emph{irredundancy witness history}.
This definition expresses that for some histories, given any execution that explains
it with queries that see only one of the updates, making one update visible
to the queries that see the other, before the visible update, will render the
resulting execution invalid (even if well-formed). This will be due to the
histories containing some query operation, for which the newly visible
operation, serialized before the already visible one and not being overridden
by it, but complementing it instead, affects the results, breaking result
validity.

Intuitively, it means that none of the two operations is redundant; having
both operations visible produces a different outcome than having just one of
them visible. The difference in outcomes may need more than one query to be
detected, and not all histories may reveal it;
thus the existential quantification in the definition. The definition
involves a pair of updates and it is about the effect of these to other
operations, not one operation to the other. (Otherwise, a pair containing an
update and a pure query which depends on it, such as a write and a read of a
single register, would be considered an irredundant pair.)
The definition involves adding the invisible update before the visible one, to
reflect the impossibility of adding an edge as phantom. Adding it after would
normally change the outcome, even for a single register.

Examples are any pair of updates, even if identical, for non-idempotent
operations, such as $(\inc(1), \inc(1))$ for a counter, $(\enq(1), \enq(1))$
for a queue, or $(\push(1), \push(1))$ for a stack; or operations
which, even if idempotent, are applied to different objects/locations/elements,
such as $(\add(x), \add(y))$ for a set, or $(\af{wr}(x, 1), \af{wr}(y, 2))$ in the
memory example. In almost all abstractions defined with sequential semantics,
not only the last update matters.

\subsection{The CLAM theorem}

We can now state the CLAM theorem. Being based on the concept of irredundant
operations, it applies to practically all abstractions, where not every
operation overrides every other operation; the single register was the only
counter-example we have found.

\begin{theorem}[CLAM]
In asynchronous distributed systems with at least two processes, any wait-free
implementation of a data abstraction with irredundant updates allows
histories that cannot be explained by any execution that simultaneously
satisfies Closed past, Local visibility, Arbitration, and Monotonic visibility.
\end{theorem}

\begin{proof}
By contradiction. Assume a system with at least two processes $i$ and $j$, and
a data abstraction having irredundant updates with a wait-free
implementation. Assume an irredundancy witness history $H$ with one update at each
process, $u_i$ and $u_j$, forming an irredundant pair, and one or more queries
at each process, subsequent to the respective update. Assume a partition which
lasts the whole history, which makes no information from $i$ reach $j$ and
vice-versa. Because the implementation is wait-free, the queries have returned
at some point before the partition heals, without depending on actions from
other processes.
This implies that the results of queries at $i$ depend only of events at $i$
(update and queries, if not pure), and the same for $j$. This means that any
execution $E$ that explains the whole history is an extension of two executions
$E_i$ and $E_j$, sharing no events, each restricted to the corresponding
process, with $E$ having the union of the respective events and program
order, visibility containing the respective visibilities and possibly some
extra edges added, and serializations which extend the respective serializations.
Assume $E$ satisfies Monotonic visibility, Arbitration, and Local visibility.
Local visibility implies that the update at $i$ is visible by all queries at
$i$, and the same for $j$.
Due to arbitration, $u_i$ and $u_j$ are ordered in the same way in both
serializations; assume WLOG that $u_i \ar u_j$.
Assume that $E$ also satisfies Closed past. Therefore, because $u_j$ is
visible to all queries at $j$, so must be $u_i$. But the history must be
explained by an execution $E_0$, similar to $E$ except that visibility is just
the union of the visibilities in $E_i$ and $E_j$, corresponding to the absence of
information propagation between $i$ and $j$ during the partition. But the
extra visibility edges in $E$ over $E_0$, from $u_i$ to all queries in $j$, to
satisfy closed past, make $E$ contradict result validity, given that $H$ is
an irredundancy witness history. This contradicts $E$ being a valid execution
which explains the history.
\end{proof}

\subsection{Comparison with the CAP theorem}

We now discuss how CLAM and CAP compare. We show that they are technically
incomparable but CLAM is practically stronger: it is stronger for any data
abstraction to which it applies.

The CAP conjecture~\cite{DBLP:conf/hotos/FoxB99} is somewhat vague; the CAP
theorem was proved~\cite{DBLP:journals/sigact/GilbertL02} for a single
linearizable register. The CLAM theorem does not apply to a single register as
it does not contain irredundant operations, so CLAM cannot be stronger than
CAP. In fact, a single register can be sequentially consistent and wait-free.

\begin{proposition}
A single register can have a wait-free implementation in an asynchronous
distributed system while being sequentially consistent.
\end{proposition}

\begin{proof}
By describing such an implementation, assuming totally ordered process ids,
consisting of:
1) processes maintain a Lamport clock~\cite{DBLP:journals/cacm/Lamport78},
updated only on write operations and message passing, and build a total
order of writes using pairs (clock, process) under lexicographic order;
2) a write is performed locally and disseminated to other processes, piggybacking
the clock value;
3) when a write message is received, it is ignored if the corresponding pair
is ordered before the local one, otherwise the value and clock is updated;
4) a read returns the locally stored value.
Considering a run of this implementation: an arbitration can be built by
ordering sequences of reads from a process after the write that was read and
before the next write in the total order of writes; if
two processes have sequences of reads for the same write, totally order those sequences
by process id; all operations arbitrated before some operation can be
considered to be visible (reads are ignored in the semantic function and
only the last write matters), and so we have a serial execution at each
process, i.e., serial consistency. Together with the arbitration, it implies
sequential consistency.
\end{proof}

\begin{remark}
This result is not practically useful. Even if a single register can have a
wait-free implementation, sequential consistency is not composable (unlike
linearizability). Sequentially consistent memory cannot be obtained from a set
of sequentially consistent registers.
\end{remark}

On the other hand, even if the CAP theorem could be proved for other data
types, in the same way considering linearizability,
C+L+A+M together are weaker than sequential consistency, which is weaker than
linearizability, making the the CLAM impossibility stronger in terms of
which consistency models it prevents.

So, in practice, CAP only says that linearizability is not possible for AP
(wait-free) systems, without shedding light on what tradeoffs can be made.
On the other hand, CLAM goes to a fine grain and states that for almost all
abstractions (those that have irredundant updates) we cannot have even the
weaker combination of properties C+L+A+M together if we want a wait-free
system.
As M is too fundamental to let go, CLAM shows what tradeoffs can be made in
what we call the CAL trilemma. We have seen examples of models where one
of C, L, or A was forgone, for which wait-free implementations are easily
obtained. We discuss these in more detail below.

\section{Taxonomy of consistency models}

We summarize the different axioms and consistency models discussed in the
article in Table~\ref{tab:summary}, organized in three groups. The first group has
axioms about well-formedness (physical realizability and serialization of
visibility), validity, and time constraints (given a logical or physical
clock). The second group has axioms used to define consistency models: the basic
axioms (monotonic visibility, local visibility, and closed past); axioms
related to inter-process propagation order (pipelined visibility and
serializations, together defining pipelining, and causal visibility and
serializations, together defining causality); and convergence axioms, namely
arbitration. In the third group we present consistency models: classic models
which satisfy what we have called serial consistency, the bottom of classic
models; replay models, which forgo closed past, variants of what we call replay
consistency; and finally prefix models, which forgo local visibility, variants of
prefix consistency.

\begin{table*}
  \caption{Summary of axioms and consistency models, in three main groups:
  well-formedness, validity, and time axioms; basic, propagation, and
  convergence axioms; serial, replay, and prefix models.}
  \label{tab:summary}
  \begin{center}
    \small
    \begin{tabular}{@{}llll@{}}
  \toprule
      Class & Criterion & Symbol & Definition \\
  \midrule
      \multirow3*{Well-formed}
      & physical realizability &
        \axiom{vis}{pr} & 
         $ a \hb b \implies b \not\po a$ \\
      & serialization of visibility &
        \axiom{ser}{vis} & $a \vis b_i \implies a \ser_i b_i $ \\
      & well-formed execution &
        \axiom{exe}{wf} & $\axiom{vis}{pr} \land \axiom{ser}{vis}$ \\
   \cmidrule{2-4}
      \multirow2*{Validity}
      & result validity &
        \axiom{res}{val} & $\re(o_i) = \sfR{o_i,\vis,\ser_i}$ \\
      & valid execution &
        \axiom{exe}{val} & $\axiom{exe}{wf} \land \axiom{res}{val}$ \\
   \cmidrule{2-4}
      \multirow2*{Time}
      & logical clock &
        \axiom{vis}{lc} & $a \vis b \implies L(a) < L'(b)$ \\
      & physical clock &
        \axiom{vis}{pc} & $a \vis b \implies P(a) < P'(b)$ \\
  \midrule
      \multirow3*{Basic}
      & monotonic visibility &
        \axiom{vis}{mon} & $a \vis b \po c \implies a \vis c$ \\
      & local visibility &
        \axiom{vis}{loc} & $a \po b \implies a \vis b$ \\
      & closed past &
        \axiom{ser}{clo} & $a \vis b_i \land c \cnot\vis b_i \implies a \ser_i c$ \\
   \cmidrule{2-4}
      \multirow6*{Propagation}
      & pipelined visibility &
        \axiom{vis}{pipe} & $a \po b \vis c \implies a \vis c$ \\
      & pipelined serializations &
        \axiom{ser}{pipe} & $a \po b \implies a \ser_i b$ \\
      & pipelining &
        \axiom{prop}{pipe} & $\axiom{vis}{pipe} \land \axiom{ser}{pipe}$ \\
      & causal visibility &
        \axiom{vis}{causal} & $a \hb b \implies a \vis b$ \\
      & causal serializations &
        \axiom{ser}{causal} & $ a \hb b \land b \not\hb a \implies a \ser_i b $ \\
      & causality &
        \axiom{prop}{causal} & $\axiom{vis}{causal} \land \axiom{ser}{causal}$ \\
   \cmidrule{2-4}
      \multirow2*{Convergence}
      & convergence &
        \axiom{res}{conv} & $\op(a) = \op(b) \land {\vis a} =  {\vis b}
        \implies \re(a) = \re(b)$ \\
      & arbitration &
        \axiom{ser}{arb} & ${\ser_i} = {\ser_j}$ \\
  \midrule
      \multirow5*{Serial models}
      & serial &
        \cons{serial} & ${\ser_i o_i} = {\vis o_i}$ \\
      & pipelined &
        \cons{pipe} & $\axiom{prop}{pipe} \land \cons{serial}$ \\
      & causal &
        \cons{causal} & $\axiom{prop}{causal} \land \cons{serial}$ \\
      & convergent causal &
        \cons{causal+} & $\cons{causal} \land \axiom{res}{conv}$ \\
      & sequential &
        \cons{seq} & $\cons{serial} \land \axiom{ser}{arb}$ \\
   \cmidrule{2-4}
      \multirow3*{Replay models}
      & replay &
        \cons{replay} & $\axiom{vis}{mon} \land \axiom{vis}{loc} \land
        \axiom{ser}{arb}$ \\
      & pipelined replay &
        \cons{preplay} & $\cons{replay} \land \axiom{prop}{pipe} $ \\
      & causal replay &
        \cons{creplay} & $\cons{replay} \land \axiom{prop}{causal} $ \\
   \cmidrule{2-4}
      \multirow3*{Prefix models}
      & prefix &
        \cons{prefix} & $\axiom{vis}{mon} \land \axiom{ser}{clo} \land \axiom{ser}{arb}$ \\
      & pipelined prefix &
        \cons{pprefix} & $\cons{prefix} \land \axiom{prop}{pipe} $ \\
      & causal prefix &
        \cons{cprefix} & $\cons{prefix} \land \axiom{prop}{causal} $ \\
  \bottomrule
    \end{tabular}
  \end{center}
\end{table*}

The way different axioms and consistency models are partially ordered, according to
strength is presented in Figure~\ref{fig:taxonomy}. It presents axioms in
boxes with a dashed border and consistency models with a solid border. 
Those which satisfy convergence have a dark background.
All models presented satisfy the fundamental monotonic visibility.
To better visualize how they compare, border color is an RGB combination of
red if satisfying Local visibility, blue if satisfying Closed past, and green
if satisfying Arbitration.
This makes it easy to see that all except those with a gray border (CAL) allow
wait-free implementations for arbitrary, non synchronization-oriented, data abstractions.
Classic models (except sequential consistency) form a column on the left side, variants of
replay consistency a column on the center, and variants of prefix consistency form a
column on the right.

It can be seen that causal prefix consistency ends up not allowing wait-free
implementations (therefore, being uninteresting), due to causality implying local
visibility. It can also be seen that convergent causal consistency is the
only model that allows wait-free implementations and where convergence is
achieved while also satisfying serial consistency, without forgoing either
local visibility or closed past; this is possible by not depending on
arbitration for convergence.

\begin{figure*}
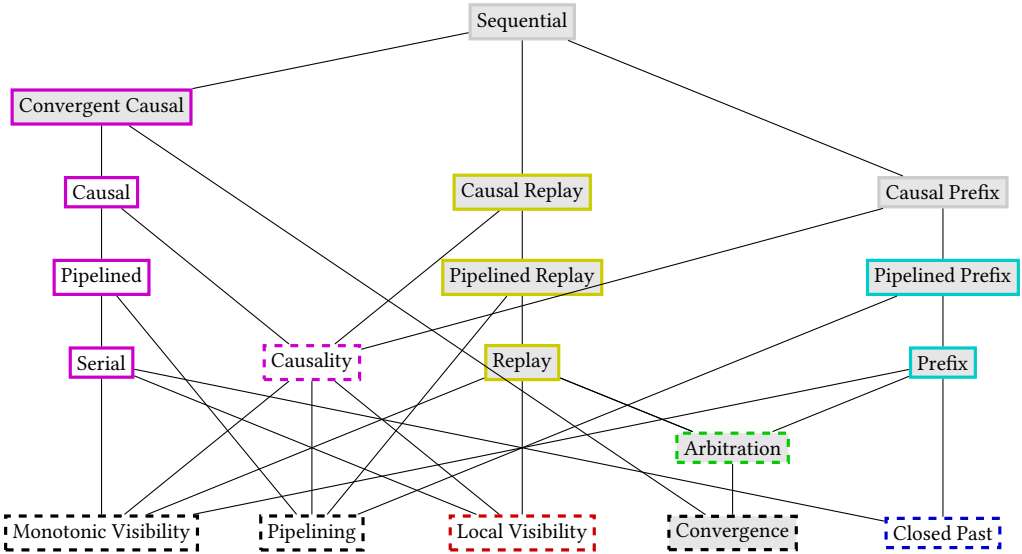

\definecolor{conv}{gray}{0.90}
\definecolor{ncon}{gray}{1.0}
\definecolor{l}{rgb}{0.8, 0.0, 0.0}
\definecolor{a}{rgb}{0.0, 0.8, 0.0}
\definecolor{c}{rgb}{0.0, 0.0, 0.8}
\definecolor{lc}{rgb}{0.8, 0.0, 0.8}
\definecolor{la}{rgb}{0.8, 0.8, 0.0}
\definecolor{ca}{rgb}{0.0, 0.8, 0.8}
\definecolor{lca}{rgb}{0.8, 0.8, 0.8}
  \footnotesize
  \setlength\h{28mm}
  \tikz \graph
  [nodes={very thick, fill=conv, draw=black}, grow down=0.4\h, branch right=\h]
  {
    { seq/Sequential [draw=lca, x=2\h] --
     cc/Convergent Causal [draw=lc] --[matching]
     {{[nodes={draw=lc, fill=ncon}]
      Causal -- Pipelined  --[matching] { Serial,  Causality [dashed]}},
      {[nodes={draw=la}]
      cu/Causal Replay  -- pu/Pipelined Replay -- Replay } --
      arb/Arbitration [x=\h, draw=a, dashed],
      {[nodes={draw=ca}]
      cp/Causal Prefix [draw=lca, x=\h] -- pp/Pipelined Prefix [x=\h] -- Prefix [x=\h]}
     }
      --[draw=none]
      { [nodes={dashed, fill=ncon}] mv/Monotonic Visibility, Pipelining,
      lv/Local Visibility [draw=l], Convergence [fill=conv], cpa/Closed Past [draw=c]}
    },
    seq -- {cu, cp},
    {Causal, cu, cp} -- Causality -- {mv, lv, Pipelining},
    {Pipelined, pu, pp} -- Pipelining,
    Serial -- {mv, lv, cpa},
    Replay -- {mv, lv, arb},
    Prefix -- {mv, cpa, arb},
    {cc, arb} -- Convergence,
  };
  \caption{A taxonomy of consistency axioms (dashed border) and models (solid
  border), partially ordered according to strength;
  with a dark background if satisfying convergence, 
  border color being a combination of red if satisfying Local visibility, blue
  if satisfying Closed past, and green if satisfying Arbitration;
  all except those with a gray border (CAL) allow wait-free implementations for
  arbitrary data abstractions.}
  \label{fig:taxonomy}
\end{figure*}

\section{Related work and discussion}

\subsection{Consistency frameworks}

Much work has been done regarding consistency models for shared memory. They
tend to have per-location constraints and several orders expressing possible
reorderings of read and write operations to ``the'' memory, and model
machine instructions like memory fences. Normally they do not discuss
convergence, which is implicitly assumed. Two good examples are the
axiomatic frameworks
by~\textcite{DBLP:journals/fmsd/Alglave12,DBLP:journals/toplas/AlglaveMT14}.
But shared memory is not the focus of this article.

Although addressing memory, but suitable to distributed systems, the framework
by~\textcite{DBLP:journals/jacm/SteinkeN04} was an important milestone, by
allowing classic memory consistency models to be obtained as a combination of
orthogonal constraints on serial views.  Being a timeless model, based on
program order and writes-to, but without visibility, it lacks a mechanism to
prevent physically impossible causality loops in a general way. It fails to prevent mutually
incompatible serializations, that do not respect physical information
propagation, in a model agnostic way.
But it unified classic specifications in an accurate way: the specifications
themselves (except for sequential consistency) were over relaxed and allow
physically impossible causality loops. We discuss how the specification for PRAM is more relaxed
than the original operational definition, in Remark~\ref{rem:pram}, and how
the specification of causal memory also allows some physically impossible causality loops, in
Remark~\ref{rem:cm-cl}.

Frameworks devoted to replicated
abstractions~\cite{DBLP:journals/ftpl/Burckhardt14,DBLP:conf/popl/BurckhardtGYZ14},
were developed with eventual consistency as a goal.
They resort to visibility and an arbitration, allowing concurrent
specifications, used in many CRDTs. But they do not have per-process
serializations, not being able to describe models without convergence.
The limitation comes from having a single global arbitration. Our model
keeps the classic per process serialization, suffering no such limitation,
allowing consistency models that do not imply convergence. Moreover even when
convergence is desired, arbitration is not always needed, or necessarily
desirable. In our framework arbitration is a criterion.

On the other hand, \textcite{DBLP:conf/ppopp/PerrinMJ16}, who point
the limitations of having visibility and a global arbitration,
generalize causal consistency beyond memory to abstract data types, but is
limited to sequential specifications. It is based on serial views, allowing
per-process serializations and the absence of convergence, but does not use
visibility, therefore not being able to restrict the operations in a view
according to visibility, nor to use visibility for defining concurrent
specifications, something important to CRDTs. The framework does not
prevent physically impossible causality loops in general: it does so for causal consistency, by
the existential assumption of a single causal order that must be respected by
all serializations, but fails to do it for the generalization of PRAM into
pipelined consistency, importing the same problem of the PRAM axiomatic specification
being overly relaxed, allowing histories with mutually incompatible
serializations. As an example, their Figure 3(i) is pipelined consistent
according to their specification, as the PRAM axiomatic specification allows,
but has a causality loop and would not be accepted by the original operational
definition of PRAM, as we discuss in Remark~\ref{rem:pram}.

Like~\textcite{DBLP:journals/jacm/SteinkeN04}, we define a set of axioms
that can be combined in different ways, to achieve a partial order of models.
Our framework is timeless but prevents physically impossible causality loops in a general way
through the \axiom{vis}{pr} well-formedness axiom. It uses both visibility and
per-process serializations, allows general data types, both sequential and
concurrent specifications, convergence as a property, achievable in different
ways, with arbitration also as a property which may or may not hold. 
It also allows a wider range of models to be defined that are relevant to
distributed systems, but not explainable by serial executions, such as the
prefix and replay consistency families of models.

Some models or frameworks have operational definitions, and either
architecture or implementation related assumptions, which hurts their
generality and make comparing models and building taxonomies harder.  An
example, the framework by~\textcite{DBLP:conf/concur/ShapiroAP16}, assumes
clients and replicas, and that operations are divided into calls, generators,
effectors and returns (somewhat like operation-based CRDTs). It allows
per-process serializations but not an independent visibility (it is derived
from the serializations), not allowing the definition of replay or prefix
consistency.

Many presentations of the subject tend to be considerably cluttered with
accidental complexity, details that are irrelevant to the essential matter,
which also frequently results in them being unnecessarily restrictive.
Examples being formally defining ADTs, their alphabets, distinguishing updates
from queries, having collections of objects with operations acting only on a
single object. Doing this and restricting visibility to relate operations on
the same
object~\cite{DBLP:journals/ftpl/Burckhardt14,DBLP:conf/popl/BurckhardtGYZ14},
not only prevents having operations that involve several objects, but also
has unfortunate consequences in terms of preventing simple and elegant
constraints, such as ${\vis} = {\hb}$ for causal consistency.

One of the more expressive frameworks, by \textcite{DBLP:conf/srds/JiangW020}, 
allows specifying models beyond serial consistency, as our own framework.
It generalizes the framework by \textcite{DBLP:journals/ftpl/Burckhardt14},
making arbitration a partial order and adding a function to the framework to
do something similar to what \textcite{DBLP:conf/ppopp/PerrinMJ16} does for
individual models: express which return values of visible operations cannot be
ignored when explaining the result of a given operation. The framework uses
sequential specifications and conflates convergence with arbitration being a
total order. The more problematic issue is that arbitration constraints both
local and remote serializations. The framework assumes possible serializations
for a given abstract execution, constrained by arbitration (not made explicit
in the sequential semantics but clearly intended, as otherwise arbitration
would play no role). This means that, to constrain local serializations to
respect program order, arbitration must respect it, i.e., ${\po} \subseteq
{\ar}$.  But doing so, serializations at other processes will also respect it,
making the framework unsuitable to models not respecting pipelined
serializations, such as serial consistency. On the other hand, it allows
expressing outcomes not possible by our framework, easily resulting in models
that do not arise in actual systems and do not suit the notion of a sequential
process, as we discuss in Section~\ref{sec:monotonic-serialization}.

We consider our framework, with the traditional per-process serialization
augmented with visibility, to be a sweet spot in allowing useful generality,
precision in describing desirable outcomes, and preventing undesirable
phenomena without the need for extra ad-hoc per-model constraints, while
keeping simplicity and intuitiveness. We detail some of these aspects below.

\subsection{Time in consistency models}

Most consistency models do not use time. Those who do, like Linearizability,
define operations as a pair of invocation and response events, placed in a
timeline, i.e., assuming a totally ordered physical time.
Frameworks that consider
time~\cite{DBLP:journals/ftpl/Burckhardt14,DBLP:journals/csur/ViottiV16} adopt
the same approach, assigning a start and a return time for each operation, and
define a \emph{returns-before} partial order according to precedence in
real-time. So, a history in these frameworks already includes time
information. Specific consistency models, such as linearizability, can be defined using a
\emph{Realtime} axiom defined in terms of the returns-before relation.

As we have discussed in Section~\ref{sec:time}, requiring that a history already
includes information about time makes it more difficult and fragile to collect
histories from actual runs of a system to be checked. We kept our framework
timeless, but discussed how information and constraints about time can be
optionally added, orthogonally.

Importantly, we allow several variants for
partially or totally ordered physical time, and also the use of logical time
to generate constraints. This allows instrumenting a system with a logical clock
and not accepting an existentially assumed execution in which $a$ is visible
to $b$, if no information could have reached $b$ from $a$, which is the case
if  the logical clock when $a$ starts is not less than the one when $b$
ends. Precedence in a totally ordered physical time is typically used to constrain
serializations but must be more conservative in pruning visibility than
constraints generated by a partially ordered logical clock. Moreover, if $a$
and $b$ are spacelike in space-time, outside each other light-cone,
enforcing a specific serialization between them
(according to totally ordered physical time) will not admit some executions
which should be perfectly valid, even if the goal is to avoid missing
information that could have been propagated between $a$ and $b$ by hidden
channels, such as in \emph{external consistency}~\cite{10.5555/910052}.

Preventing physically impossible causality loops and out-of-thin-air values is one
major concern for consistency models, but most timeless models for
distributed systems have failed to prevent them.
As an example, the formalism used for causal
memory~\cite{DBLP:journals/dc/AhamadNBKH95} is timeless, as ours, 
but lacks an axiom to prevent some non obvious causality loops: a technical
issue caused by multiple writes of the same value to the same location,
causing a possible discrepancy between a causality order existentially assumed
and the real cause of a read, allowing histories with causality loops,
which should have been rejected. Understanding whether such outcomes are
possible in a specific model, and adding ad hoc axioms to prevent them is
non trivial.

A significant feature of our framework is that it abstracts away time, while
having physically impossible causality loops disallowed by the physical realizability axiom
$\axiom{vis}{pr}$ of well-formed executions.
This frees specific consistency models from individually having to worry about
causality loops; they just need to assume well-formed executions. Notably,
physical realizability allows cyclic visibility dependencies, as long as program
order is not involved. This allows synchronization-oriented abstractions,
something not possible with the more standard acyclic happens-before.

\subsection{Causal consistency}

Causal consistency has been introduced
by~\textcite{DBLP:journals/dc/AhamadNBKH95}, in the context of memory, as
causal memory, and generalized beyond memory
by~\textcite{DBLP:conf/ppopp/PerrinMJ16} to data types having a sequential
specification. Our framework generalizes it beyond sequential specifications,
allowing also concurrent specifications. 

There is a discrepancy between the specification of causal memory, lacking a
relevant axiom, and its stated goal, namely being stronger than PRAM (it is
stronger than the PRAM axiomatic specification, itself lacking an axiom, but
not the PRAM operational definition). As we discussed in the context of
pipelined consistency, our framework overcomes this issue in a model
independent way, generalizing causal memory while better fitting the
intention, by discarding physically impossible histories. It describes more
accurately the possible outcomes when instantiated for memory than the causal
memory specification itself.

We have defined Causality as a common denominator for several variants of
causal consistency (standard, causal prefix, and causal replay).  It is the
natural generalization of the definition by~\textcite{DBLP:journals/ftpl/Burckhardt14}
to models where convergence is not implied. In that work,
for eventually consistent systems, with arbitration being assumed, causality
corresponds to our causality+arbitration, i.e., causal replay consistency, and
causal consistency corresponds to causal replay consistency (for valid
executions) plus the liveness property of eventual visibility. As we did for
time, we leave liveness outside the framework.

\emph{Weak Causal Consistency (WCC)}, introduced
by~\textcite{DBLP:conf/ppopp/PerrinMJ16} as a common denominator for different
variants, is considerably weaker than our causality, as it allows different
operations, even from the same process, to be explained by different
serializations of the causal past; i.e., a per-operation serialization is
allowed (not a per-process serialization, as in our framework) as long as it
respects happens-before.

\textcite{DBLP:conf/popl/BouajjaniEGH17} provide a definition for
causal consistency which is similar to WCC, in
which outcomes are also explained by a per-operation serialization.
Such definitions allow histories such as the one in
Figure~\ref{fig:per-operation-serialization} where, given a pair of concurrent
writes, a sequence of reads can keep returning any of those written values,
alternating in an apparently random fashion along time, without
ever stabilizing. No classic model, even the quite weak Slow Memory, or
the ones introduced here allow such range of outcomes. Such definition has no
resemblance to the standard definitions of causal consistency, whether for
memory~\cite{DBLP:journals/dc/AhamadNBKH95} or its generalization to
sequential data types~\cite{DBLP:conf/ppopp/PerrinMJ16} or concurrent data
types as in this article.  These definitions do not match what can be
expected from a sequential process incorporating concurrent updates.
Moreover, reusing the same term ``causal consistency'' for a model so
substantially different and weaker than the
traditional~\cite{DBLP:conf/fsttcs/RaynalS95,DBLP:journals/jacm/SteinkeN04}
usage of the term has the potential to cause considerable confusion.  When
discussing examples of histories satisfying causal consistency, \textcite[Fig.
3]{DBLP:conf/fsttcs/RaynalS95} explicitly reject such patterns of
alternating values.  When generalizing from memory to arbitrary sequential data
types, \textcite{DBLP:conf/ppopp/PerrinMJ16} avoid this confusion by introducing
a different term (WCC) for this weaker model.

\begin{figure*}
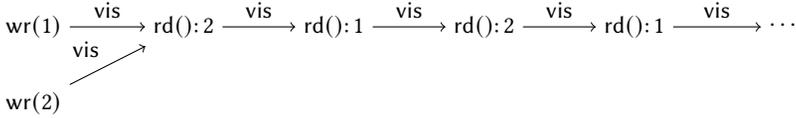

  \small
  \begin{boxes}{1}
  \let\wr=\relax
  \let\rd=\relax
  \auxfun{wr}
  \auxfun{rd}
  \begin{tcolorbox}
      \tikz \graph [edge label=${\vis}$, grow right=20mm] {
        a/$\wr(1)$ -> b/$\rd():2$ -> c/$\rd():1$ -> d/$\rd():2$ -> e/$\rd():1$
        -> f/{$\cdots$} ;
        g/$\wr(2)$;
        g -> b;
      };
  \end{tcolorbox}
  \end{boxes}
  \caption{History only explainable by a per-operation serialization, allowing
  an infinite sequence of apparently random values from $\{1, 2\}$ to be read over
  time.
  }
  \label{fig:per-operation-serialization}
\end{figure*}

\subsection{The rationale for per-process serializations: monotonic serialization}
\label{sec:monotonic-serialization}

Regardless of its reasonableness, the above weak definition of causal
consistency raises the question of how expressive a consistency framework
needs to be. In particular, whether the choice of a per-process serialization
in our framework is not overly restrictive. We settled on a per-process
serialization, complemented by visibility, as we found no reason, or example,
to justify the range of outcomes possible by a per-operation serialization.

\textcite{DBLP:conf/popl/BouajjaniEGH17} give as motivation for per-operation
serialization that a process can ``change its mind'', to allow speculative
execution and rollbacks, motivated by the desire to achieve convergence. But
we have not found a single example of an actual system in which processes
``change their minds'' along time regarding how two operations are serialized
in a given process. They either:

\begin{itemize}
  \item Do not care about convergence and each process serializes operations
    as they become visible, satisfying serial consistency.
  \item Achieve convergence through arbitration,
    making local operations immediately visible, but possibly serializing a
    remote operation that arrives later before already visible ones, forgoing
    closed past, requiring (conceptual) re-executing of operations, satisfying
    replay consistency. This is the case in the classic
    Bayou~\cite{DBLP:conf/sosp/TerryTPDSH95}, or the more recent
    OpSets~\cite{DBLP:journals/corr/abs-1805-04263} and
    ECROs~\cite{DBLP:journals/pacmpl/PorreFPB21}.
  \item Achieve convergence through arbitration,
    by delaying the visibility of operations, even local ones, until they are
    sure no operation may arrive that would be serialized before. They forgo
    local visibility but do not require re-execution, building an arbitration
    and making visible a prefix of the arbitration. This is the case of prefix
    consistency.
  \item Converge without needing an arbitration, by using a concurrent
    specification. This is the case of many CRDTs.
\end{itemize}

So, in all systems that we are aware of using sequential specifications, the first
moment two operations are made visible to a process, their order in the
process has been decided, remaining immutable.
Processes do not ``change their mind'' but either delay visibility or insert new
operations into the past, when they resort to arbitration for convergence.

Moreover, perpetually alternating values as in
Figure~\ref{fig:per-operation-serialization} do not match the notion of a
sequential process, even in a concurrent setting. Models for sequential
processes and sequential specifications expect that when the effects of all
concurrent operations have been incorporated, the process has one of several
possible states. It is expected that the current state results from a sequence
of updates leading to that state; the state will explain possible results of
future queries. Such alternating values do not occur in deterministic
specifications. Given that one goal of a consistency model is to accurately
define the range of desired or possible outcomes of a given system, such
consistency criteria that allow these alternating values fail in that goal.

Our framework explains desired/actual outcomes in a precise way. It implicitly enforces
a \emph{Monotonic Serializations} axiom that could be defined in a framework
with per-operation serialization: that as the process
evolves and more operations become visible (monotonic visibility) the new
operations add pairs to the current serialization but do not contradict the
current pairs. Monotonic visibility and monotonic serializations can be seen
as the more fundamental axioms for a desirable consistency model.

\subsection{Eventual consistency and convergence}

\textcite{DBLP:journals/cacm/Vogels09} popularized a definition of eventual
consistency (EC) as eventual convergence when updates stop being issued.
This definition is not very useful, as updates may not stop in many/most systems.
This liveness property is just one consequence of what EC systems ensure, but
they must ensure something stronger, which implies this property: different
replicas must keep continuously converging as operations arrive, so that
when/if update operations stop being issued replicas will eventually converge.

\textcite{DBLP:conf/sss/ShapiroPBZ11} defined \emph{strong eventual
consistency} (SEC), which guarantees that all updates will eventually become
visible everywhere (eventual visibility) and that processes that see the same
set of updates have equivalent state, regardless of the order in which they become
visible (\emph{strong convergence}). This definition also mixes safety and
liveness properties.

As we discuss elsewhere~\cite{10.1145/3695249}, SEC was an unfortunate choice
of terminology, possibly causing some confusion, due to the use of the word
\emph{strong}, normally used for strong consistency models, such as sequential
consistency and linearizability. More importantly, the original guarantees
that eventually consistent systems provide, when eventual consistency was
introduced by~\textcite{DBLP:conf/pdis/TerryDPSTW94}, already includes SEC.
In that paper, EC systems must have mechanisms to ensure two properties:
\begin{itemize}
  \item \emph{total propagation}: each update is eventually propagated
    everywhere, by some anti-entropy mechanism (i.e., eventual visibility);
  \item \emph{consistent ordering}: non-commutative updates are applied in the
    same order everywhere, which implies the \emph{strong convergence}
    property of SEC.
\end{itemize}
As described by~\textcite{DBLP:conf/sosp/TerryTPDSH95}, also addressing EC,
these properties mean that ``all servers eventually receive all Writes via
the pair-wise anti-entropy process and that two servers holding the same set
of Writes will have the same data contents''. So, from these descriptions, EC
systems already satisfy SEC. Eventual convergence when updates stop is a
mere consequence of these properties, but not enough in itself.

A definition of EC similar to the original meaning, ensuring (strong)
convergence was also used by~\textcite{DBLP:journals/jpdc/RohJKL11}.
As originally introduced, what is ``eventual'' in EC systems is operation
visibility (delivery); replicas that have delivered the same set of updates
have converged.

\begin{remark}
It would be difficult to achieve convergence otherwise, if it had not been
achieved for the same set of visible events. If that was not the
case, \emph{when} would replicas converge? I.e., as result of processing
\emph{which} other events, namely if there were no more events to be acted
upon? An implementation cannot know if updates have stopped being issued.
In a framework as ours, the semantic function that defines the result of a
query uses the visible operations, no more and no less.
``Visible'' means semantically visible to the process using the
abstraction; in implementation terms, information may have been received by
the process but kept in some buffer and not yet delivered to the abstraction,
therefore, still not visible.
It would be possible to resort to implementation level events, on which
convergence depends, but that defeats the purpose of specifying a shared
object while abstracting away from implementation details.
If some explicit action, like ``reconcile'', is necessary for convergence,
that becomes part of the API and used in the semantic function.
So, SEC is not something that implies some extra cost over the baseline,
but the natural condition for EC systems. The term SEC itself should be deprecated in
favor of using EC, as originally introduced: the combination of eventual
visibility and convergence as we define it.
\end{remark}

Defining convergence as we do, intending the ``strong'' variety is not
original. It was already done, e.g.,
by~\textcite{DBLP:journals/tochi/SunJZYC98} in the context of
collaborative editing: ``When the same set of operations have been executed at
all sites, all copies of the shared document are identical''.
What is original in our approach is using convergence as an axiom for consistency
models and taxonomy, rather than eventual consistency as
usual~\cite{DBLP:journals/csur/ViottiV16}. Discarding the liveness part of
EC, as we do, is the right choice to define consistency models that are
suitable to both finite and infinite histories.

Models defined using EC, specially the weak variant, end up not being useful
for finite histories. An example is \emph{Update
Consistency}~\cite{DBLP:conf/ipps/PerrinMJ15}, which accepts most finite
histories as valid (as invalid results from an arbitrary finite set of queries
can be ignored). We see using EC to base consistency models on as a ``dead end''.

\subsection{Convergence, causality, and wait-freedom}

Models that satisfy an arbitration, such as sequential consistency, or the
non-serial models that we introduced (replay and prefix consistency) satisfy
convergence. But arbitration is just one possible means to achieve
it. Convergence can be achieved in a wait-free way using concurrent
specifications, common in CRDTs, but we point out that even some data types
with sequential specifications can achieve wait-free convergence without
arbitration, namely data types in which update operations are commutative.

As an example, consider the execution in
Figure~\ref{fig:convergence-examples-counter} for a counter object, satisfying
convergence. Each process sees the two increments in different orders, but
when both increments are visible both queries return the same result. The
execution is also causally consistent, and such a counter can have a wait-free
implementation. Another example is a grow-only set, without a remove operation.

\textcite{DBLP:conf/ppopp/PerrinMJ16} present two incomparable models:
Causal Consistency, which does not ensure convergence; and
Causal Convergence, which allows wait-free implementations but is not stronger
than Causal Consistency. But the result from~\textcite{DBLP:conf/ipps/PerrinMJ15},
that we cannot have a wait-free implementation that is both pipelined consistent (let
alone causal consistent) and eventually consistent was proven for a set data
type with a sequential specification.  It does not mean that other
abstractions cannot be implemented as wait-free. The result also used a definition of
eventual consistency only relevant to infinite histories.
An analogous result, replacing eventual consistency with convergence, to be
useful for finite histories, also would not imply that no data type can be
implemented as wait-free while being simultaneously pipelined consistent and
convergent.

As the examples above illustrate, we can achieve both causal consistency and
convergence in several cases. The definition of causal convergence by
\textcite{DBLP:conf/ppopp/PerrinMJ16}, which requires an arbitration, is
both relaxed, as it does not ensure causal consistency, but also unnecessarily
restrictive, as convergence can be achieved by different means other than
arbitration. It corresponds
to our causal replay consistency. Our definition of convergent causal
consistency stresses that there is a model that satisfies both causal
consistency and convergence. Causal consistency allows different processes to
see different serializations, as long as they respect causality. Convergence
can be added on top not always requiring arbitration, allowing wait-free
implementations of several abstractions. Only those that require arbitration
for convergence (and have irredundant updates, as is usually the case), cannot
be wait-free, as the CLAM theorem shows.

One model that combines causal consistency and convergence is the
so called \emph{Causal+}: causal consistency with convergent conflict
handling.
Causal+ has been defined in the context of a key-value store with put/get
operations (or equivalently, memory), where a get depends only on the values 
from the most recent concurrent puts, by using a commutative and associative
conflict handling function. This corresponds to a concurrent specification in
our model.
Our convergent causal consistency can be seen
as the generalization of Causal+ beyond memory / key-value stores to arbitrary
data types.

But actual so called Causal+ systems, namely COPS, introduced
together with Causal+, may act differently. COPS performs a so called
last-writer-wins, where computing the result of a read does not depend purely
on the values written themselves, but on how updates are ordered.
It builds a total order of puts and applies or ignores an update to a given
key according to how it is ordered relatively to the update that wrote the
currently stored value.

But from the CLAM theorem, we cannot have general wait-free abstractions that
provide both serial consistency (let alone causal consistency) and
arbitration. In fact COPS allows histories that cannot be explained as
causally consistent, considering the sequential specification of memory.
An example is shown in Figure~\ref{fig:lww-convergence} in which the value of
$y$ converges to a value $v$, either $2$ or $4$. Regardless of to which value
$y$ converges, depending on how the writes to $y$ are arbitrated, one of the
processes cannot explain the history by a serial execution that satisfies
causality. It can, however, be explained by an execution forgoing closed past,
which allows serializing updates not yet received before already visible ones,
while hiding them from the semantic function.

\begin{figure*}
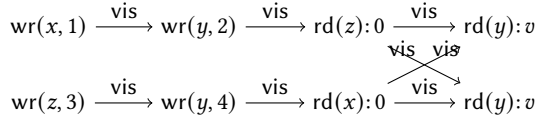

  \small
  \begin{boxes}{1}
  \let\wr=\relax
  \let\rd=\relax
  \auxfun{wr}
  \auxfun{rd}
  \begin{tcolorbox}
      \tikz \graph [edge label=${\vis}$, grow right=20mm] {
        a/{$\wr(x,1)$} -> b/{$\wr(y,2)$} -> c/{$\rd(z):0$} -> d/{$\rd(y):v$} ;
        e/{$\wr(z,3)$} -> f/{$\wr(y,4)$} -> g/{$\rd(x):0$} -> h/{$\rd(y):v$} ;
        c -> h;
        g -> d;
      };
  \end{tcolorbox}
  \end{boxes}
  \caption{Convergent execution through last-write-wins. Whatever the value of
  $v$, either $2$ or $4$, one of the processes cannot explain the outcome as a
  serial execution satisfying causality. Transitive visibility, only direct
  edges drawn.}
  \label{fig:lww-convergence}
\end{figure*}

So, formally Causal+ corresponds to our convergent causal consistency for the case of
memory, but several systems that are claimed to satisfy Causal+, namely COPS,
do not satisfy it but something else, like causal replay consistency.

A more suitable encompassing definition of Causal Convergence, to contemplate
all the different cases that arise in practice, relaxing the one
from~\textcite{DBLP:conf/ppopp/PerrinMJ16}, is to simply define it as the
conjunction of causality and convergence from our framework:

\[
  \text{Causal Convergence} \defeq
\axiom{prop}{causal} \land \axiom{res}{conv}.
\]

\subsection{Non-serial models: replay and prefix consistency}

Pipelined replay consistency corresponds to the safety conditions of
\emph{strong update consistency}~\cite{DBLP:conf/ipps/PerrinMJ15}, ignoring its liveness
requirements. Replay consistency by itself does not enforce pipelining, namely
pipelined visibility.
Apart from the minor differences regarding pipelining, we adopted the
different term ``replay'' instead of ``strong update'' for several reasons: to
avoid confusion because we do not require the liveness conditions; we dislike the
``strong'' term for weak models; the same work already defines an ``update
consistency'' variant which accepts most finite histories as valid (as
invalid results from an arbitrary finite set of queries can be ignored); but
most importantly because ``update'' does not convey its essence, that
it involves (conceptual) re-execution of operations. Prefix consistency also
builds a total order of updates, while providing different guarantees.

Prefix consistency in our framework is similar to \emph{Consistent
Prefix}~\cite{DBLP:journals/cacm/Terry13}, although this criteria also includes
satisfying pipelining, i.e., what we call pipelined prefix consistency.
Consistent prefix was only described informally.  A related model is
\emph{Monotonic Prefix Consistency}~\cite{DBLP:conf/forte/GiraultGGHS18}, but
this one uses a non-traditional definition of trace where writes are unordered
and unrelated to reads, being suitable to blockchains, but not for the
traditional notion of process used in almost all consistency models.
``Monotonic'' there is about monotonic visibility, a fundamental axiom which
we do not include in the name of any model.

\subsection{Sequential versus concurrent specifications}

All classic models and most generalizations tend to focus on sequential
specifications. 
\textcite{DBLP:conf/ipps/PerrinMJ15,DBLP:conf/ppopp/PerrinMJ16} argue that
concurrent specifications are limited to systems ensuring convergence and they
may be as complicated as the implementations themselves.
While the frameworks
by~\textcite{DBLP:journals/ftpl/Burckhardt14,DBLP:conf/popl/BurckhardtGYZ14},
resorting to visibility and arbitration, are limited to convergent models, our
framework allows both visibility and per-process serializations, and both
concurrent and sequential specifications.

We argue that there are many examples of simple and elegant concurrent specifications,
substantially less complicated than the implementation. Even a counter allows
a concurrent specification as trivial as a sequential one, with
implementations of replicated counters being considerably complex if aiming
for scalability~\cite{DBLP:journals/dc/AlmeidaB19} or to allow embedding the
counter in other data structures~\cite{DBLP:conf/eurosys/WeidnerA22}. An
ORSet has an elegant specification (see
Table~\ref{tab:concurrent-specs}), while deriving optimized implementations
that avoid the need for tombstones requires some
care~\cite{DBLP:journals/corr/abs-1210-3368}.

Concurrent specifications are indeed important for CRDTs, which have been
successfully used in the industry. They allow achieving convergence without
the need for arbitration. From the CLAM theorem and CAL trilemma, they are the
only way to achieve wait-free implementations for arbitrary abstractions while
satisfying both local visibility and closed past. With sequential
specifications, there is the need to forgo one of them, either delaying
visibility of local updates or requiring conceptual re-execution,
which removes some of the advantages of using sequential specifications in terms
of reasoning about a program. Moreover, synchronization-oriented abstractions,
such as a barrier or a consensus data type, need a concurrent specification.

Our framework allows specifying a wide range of outcomes from concurrent
shared abstractions, suitable to distributed systems.
Other models for objects with non sequential specifications, such as
set-linearizability types~\cite{DBLP:conf/podc/Neiger94} allow very limited outcomes.
Even the more expressive
interval-linearizability~\cite{DBLP:journals/jacm/CastanedaRR18}, which
allows explaining the outcome of an operation as several transitions, suffers
from the same essential limitation resulting from using a totally ordered global time
model. In all these variations, when an operation finishes, its
effects have reached the entire system, and are seen by any operation that
subsequently starts. In our framework, they imply:
\[
  P'(a) < P(B) \implies a \vis b
\]
Although this can be somewhat reasonable for shared-memory systems with a very
small spatial span, it is unreasonable as a base assumption for general
distributed system, as it implies strong consistency models. Many shared
abstractions for systems with large spatial spans, either aiming for high
availability or low latency do not meet it.

Our visibility simply ensures atomic visibility of effects (the A in ACID)
leaving any ordering of even permanence of visibility open. Visibility
may propagate asynchronously over the system after an operation concludes.
Together with visibility cycles allowed physical realizability, it allows
explaining a wide range of histories where some processes have local
interactions, some operations mutually seeing the effects of each other, but
without ``immediate global reach''.

Consider a set variant which, like the
write\_snapshot~\cite{DBLP:journals/jacm/CastanedaRR18},
has a single operation which is both an update and a query.
In this set data type, the $\add$ adds an element to the set and also returns
the resulting set:
\[
  \sfC{\add^o(s, v)} = \{v\} \union \{ x | \add(s, x) \vis o \}
\]
And consider the execution in Figure~\ref{fig:set-val} for this set.
This history does not satisfy interval-linearizability (however it is
augmented with invoke-response time information). The problem is that the
insertion of 1 and 3, which are not both seen by the first add of processes
$i$ and $k$ must have occurred at different transitions, one before the other.
Therefore, either a returned set containing 1 must also contain 2 or vice
versa, which is not the case.
But this outcome would naturally result from a simple distributed
implementation where each process stores a copy of the set and disseminates
the inserted value upon an add, but ``waits for a bit'' before returning,
processing any message that happens to arrive.

\begin{figure*}
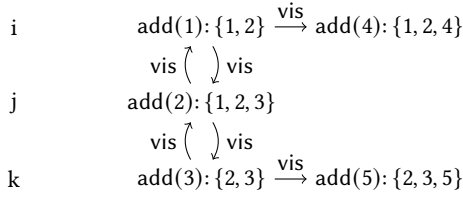

  \small
  \begin{boxes}{1}
  \let\wr=\relax
  \let\rd=\relax
  \auxfun{wr}
  \auxfun{rd}
  \begin{tcolorbox}
      \tikz \graph [edge label=${\vis}$, grow right=25mm] {
        i -!- a/{$\add(1):\{1, 2\}$} -> d/{$\add(4):\{1, 2, 4\}$};
        j -!- b/{$\add(2):\{1, 2, 3\}$};
        k -!- c/{$\add(3):\{2, 3\}$} -> e/{$\add(5):\{2, 3, 5\}$};
        a ->[bend left] b;
        b ->[bend left] a;
        b ->[bend left] c;
        c ->[bend left] b;
      };
  \end{tcolorbox}
  \end{boxes}
  \caption{Execution for concurrent set data type where $\add$  also returns
  resulting set. Monotonic but non-transitive visibility, inferred edges not drawn.}
  \label{fig:set-val}
\end{figure*}

Like for sequential specifications (e.g., memory), concurrent specifications
like this can have many different implementations satisfying different
consistency models, like Serial, Pipelined, Convergent Causal or Sequential
consistency, most of which would not satisfy interval-linearizability.  Our
framework is general as it does not tie specifications to a a specific
consistency model.

\subsection{The CAP theorem}

The \emph{Strong CAP Principle} conjectured
by~\textcite{DBLP:conf/hotos/FoxB99} was stated in terms of serializability.
In the same work they mention there is a \emph{Weak CAP Principle} about finer
grain guarantees, but leave it vague and without formalization.
The CAP theorem~\cite{DBLP:journals/sigact/GilbertL02} is about not being
possible to have a wait-free linearizable register. (It does not prevent a
wait-free sequentially consistent register.)

CAP is frequently mentioned as the trilemma of choosing between CA, CP, and
AP. But partitions occur beyond our control, even in surprising
scenarios~\cite{DBLP:journals/cacm/BailisK14} where one would not expect them
possible. So, CA is not really an option: any distributed application must have some
defined behavior for \emph{when} a partition occurs (and it will, sooner or later).

So, CAP only states that we can choose to forgo Availability to ensure
strong Consistency, or vice-versa, i.e., a choice between CP and AP. If one desires an
always-available system CAP does provide actionable tradeoffs than can be
made in the context of AP systems, namely what weaker consistency models can
be ensured while remaining available.

Recognizing the limitations of the practical impact of CAP,
\textcite{DBLP:journals/computer/Abadi12} introduces PACELC as a way to better
classify practical systems, according to the tradeoff between consistency and
latency. PACELC refers to: under Partitions how does a system trades-off Availability and
Consistency, Else (in the absence of partitions) what is the tradeoff between
Latency and Consistency.
Also addressing CAP, \textcite{DBLP:journals/corr/Kleppmann15} discusses the impact
of network latency in consistency models.

But these works do not provide axioms that can be used as building blocks for
consistency models and used in a more refined CAP-like theorem, as our CLAM
theorem does. Without the basic axioms that we introduced, the CLAM theorem
cannot even be formulated. Our Closed past, Arbitration, Local visibility, and
Monotonic visibility axioms, and the CLAM theorem allow a more precise
characterization of possible tradeoffs in wait-free (available) systems.
The CAL trilemma clearly points to either: 1) achieving convergence and
serial consistency if not depending on arbitration; 2) having local visibility
and arbitration for convergence, but forgoing closed past (replay
consistency); 3) forgoing local visibility instead, while having closed past
and arbitration for convergence (prefix consistency). Previously to our work,
the two families of models corresponding to the second and third options,
relevant for practical AP systems, had not yet been formally characterized or
made part of current taxonomies~\cite{DBLP:journals/csur/ViottiV16}.

\section{Conclusions}

While axiomatic models are better than the operational, to compare models and
build taxonomies, it may be surprisingly difficult to see whether possible
outcomes are really what we intend. This has been the case for PRAM and causal
memory, whose axiomatic specifications allow unreasonable outcomes, only
explainable by physically impossible causality loops. For PRAM the discrepancy
between the original operational definition and the axiomatic specification
used everywhere seems to have passed unnoticed all this time. For causal memory,
an issue concerning using several identical writes was spotted but the
essential issue was not clearly identified.

We have introduced a framework for consistency models in distributed
systems based on timeless histories. It allows using time information,
whether from logical or physical clocks, whether partially or totally ordered,
as optional, orthogonal, constraints, keeping models themselves timeless.
But we have defined a physical realizability axiom that prevents causality loops
that ``bring information from the future'' but allows synchronization-oriented
abstractions, such as barriers, to be specified, even when no time information
is available. This frees specific models from having to introduce ad hoc
axioms to overcome the problem, something error prone and easily missed, as
the PRAM and causal memory examples show.

Our framework achieves generality simultaneously in several dimensions: beyond
memory to general data types; beyond sequential specifications to both
sequential and concurrent, even allowing synchronization-oriented
abstractions through partial specifications; beyond serial executions to
models where local visibility or closed past does not hold. It does
this by combining per-process serializations, used in classic axiomatic
models, with visibility, used in models devoted to eventual consistency.

Like the framework by~\textcite{DBLP:journals/jacm/SteinkeN04} based in serial
views, where models are build from a combination of several properties, we
defined three basic axioms that can be composed to obtain other models:
monotonic visibility (the most essential axiom), local visibility, and closed
past. We also defined serial consistency, stronger than the three basic axioms
combined, as the generalization for our framework of local consistency,
satisfied by all classic models where each process explains the outcome as
some serial execution of its operations and the effects of remote operations
in some order.
Together with physical realizability, closed past is the other more original
axiom that we introduced. It allows classifying a model as allowing finality
of execution versus requiring (conceptual) re-execution of operations.

We also defined axioms related to the order of inter-process information
propagation: pipelining and causality. These can be combined with other axioms
to obtain variants of pipelined and causal consistency. There has been
considerable confusion about the meaning of causal consistency, not helped by
the pervading overloading of the term for significantly different definitions.
Our generalization both beyond memory and sequential specifications preserves
the original intent when instantiated for sequential specifications and
memory (it coincides with causal memory with its problem corrected).

We have introduced convergence as a stand-alone safety property for consistency
models. Current models and taxonomies have been based on eventual consistency,
which combines both safety and liveness properties, which is not suitable to
define consistency models that can be applied to finite histories. Convergence
is different than other criteria in that it requires universal quantification
over executions to classify a history as convergent, as opposed to the
standard existential quantification for other criteria.
We have introduced arbitration as a criteria to be used for consistency
models, implying convergence.
This is unlike most convergence oriented models that assume the existence of
an arbitration, and do not allow models without convergence. We also point out that
arbitration is just one possible means to obtain convergence. Current
definitions that resort to arbitration, such as causal convergence, are
unnecessarily restrictive.

The more interesting result was how the basic axioms that we have defined can be
used to characterize the space of consistency models that allow wait-free
distributed implementations, i.e., the design space of highly
available partition tolerant systems. In this respect, the CAP theorem does
not provide any useful information about what kind of tradeoffs we can make in
those systems.  We have formulated and proved the CLAM theorem for
asynchronous distributed systems, which essentially says that any wait-free
implementation of practically all data abstractions cannot simultaneously
satisfy Closed past, Local visibility, Arbitration, and Monotonic visibility.
Given that monotonic visibility is too fundamental to forgo, this results in
the CAL trilemma: which of Closed past, Arbitration, or Local visibility to
forgo. This trilemma clarifies possible tradeoffs: either we do not rely on
arbitration for convergence, and be in the family of serial consistency,
achieving up to convergent causal consistency; or we forgo closed past, and
achieve up to causal replay consistency; or we forgo local visibility and
achieve up to pipelined prefix consistency. A new taxonomy shows how the
axioms that we have introduced can be combined to obtain and compare different
consistency models, clearly showing the tradeoffs resulting from the CAL
trilemma.

\printbibliography

\end{document}